\documentclass[letterpaper]{article} 
\usepackage{aaai23}  
\usepackage{times}  
\usepackage{helvet}  
\usepackage{courier}  
\usepackage[hyphens]{url}  
\usepackage{graphicx} 
\usepackage{natbib}  
\usepackage{caption} 
\usepackage{algorithm}
\usepackage{algorithmic}
\usepackage{multirow}
\usepackage{pifont}
\usepackage{newfloat}
\usepackage{listings}
\DeclareCaptionStyle{ruled}{labelfont=normalfont,labelsep=colon,strut=off} 
\floatstyle{ruled}
\newfloat{listing}{tb}{lst}{}
\floatname{listing}{Listing}
\nocopyright

\usepackage{amsmath,amsthm,amssymb,cleveref,enumitem,mathtools,pgfplots,tikz}

\allowdisplaybreaks

\newtheorem{theorem}{Theorem}[section]
\newtheorem{lemma}[theorem]{Lemma}
\newtheorem{proposition}[theorem]{Proposition}
\newtheorem{corollary}[theorem]{Corollary}
\newtheorem{example}[theorem]{Example}
\newtheorem{definition}[theorem]{Definition}
\newtheorem{reduction}{Reduction Rule}

\title{Fairness Concepts for Indivisible Items with Externalities}
\author{
Haris Aziz,\textsuperscript{\rm 1}
Warut Suksompong,\textsuperscript{\rm 2}
Zhaohong Sun,\textsuperscript{\rm 3}
Toby Walsh\textsuperscript{\rm 1}
}
\affiliations{
    \textsuperscript{\rm 1}School of Computer Science and Engineering, University of New South Wales, Australia\\
    \textsuperscript{\rm 2}School of Computing, National University of Singapore, Singapore\\
    \textsuperscript{\rm 3}AI Lab, CyberAgent, Japan
}

\begin{document}

\maketitle

\begin{abstract}
We study a fair allocation problem of indivisible items under additive externalities in which each agent also receives utility from items that are assigned to other agents. This allows us to capture scenarios in which agents benefit from or compete against one another. We extend the well-studied properties of envy-freeness up to one item (EF1) and envy-freeness up to any item (EFX) to this setting, and we propose a new fairness concept called general fair share (GFS), which applies to a more general public decision making model. We undertake a detailed study and present algorithms for finding fair allocations. 
\end{abstract}

\section{Introduction}

Fair allocation of indivisible items is an active field of research within computer science and economics \citep{BrTa96a,BCM15a,Thom15a}. The general problem is to allocate the items among the agents so as to satisfy certain fairness criteria. For example, one important fairness concept is envy-freeness, which stipulates that no agent wants to swap her bundle with another agent's bundle. The field has witnessed several new solution concepts, algorithms, and applications. 
    
In most of the work on fair allocation, agents are assumed to derive value \emph{only} from the set of items allocated to them. In this paper, we consider a significantly more general model in which an agent's value for an allocation may depend both on the agent's own bundle as well as on  the bundles of items given to other agents. The latter aspect is referred to in the economics literature as \emph{externalities}. Whereas the theory of fair allocation has progressed tremendously, the topic is relatively less developed when externalities are involved in the valuations of the agents. 

Externalities in agent preferences are present in many real-world scenarios. When resources are allocated among agents, an agent may derive positive value from resources given to the agent's friend or family member because the agent has access rights to the resource. Positive externalities can also capture settings where agents are divided into groups and each agent receives the same utility whenever some agent in the group is allocated an item, and no utility when an agent outside the group is allocated the item. Likewise, negative externalities can arise in various resource allocation settings. For example, when dividing assets among conflicting groups, the allocation of a critical asset to another group may hamper one group's functionality. Yet another example of negative externalities is the case of sport drafts, where a team may incur negative value if a valuable player is given to a competing team.

Although externalities have been considered in some prior work on resource allocation problems, the focus was on either allocation of \emph{divisible} resources \citep{SPZ13a,LZZ15a} or concepts based on maximin fair share~\citep{SSG21a}.
In this paper, we revisit important fairness concepts such as envy-freeness and consider suitable relaxations in the context of indivisible items under externalities.

Since there may not exist an envy-free allocation in general, much of the recent research has focused on relaxations of envy-freeness by removing one or more items from consideration, e.g., \emph{envy-freeness up to one item (EF1)} \citep{LMMS04a,Budi11a}. \citet{CKM+19a} proposed a stronger concept than EF1 called \emph{envy-freeness up to any item (EFX)}. The intuition is that if agent $i$ envies agent $j$'s assignment, then the envy should be eliminated when any item is removed from $j$'s assignment. \citet{ACIW19a} generalized EFX to the setting of goods and chores (i.e., negative values) when there are no externalities. The difference is that the envy from $i$ towards $j$ can be eliminated by removing agent $i$'s least preferred good from $j$'s bundle, and also by removing agent $i$'s favorite chore from $i$'s own bundle.
 
For allocation problems under externalities, envy-freeness needs to be carefully extended. 
When we consider fair allocation of goods without externalities, if agent $i$ envies agent $j$'s assignment, then removing any item from agent $j$'s bundle decreases the envy. However, this no longer holds when externalities exist. For instance, assume that agent $i$ receives value $5$ when item $a$ is assigned to agent $i$ and receives value $10$ when $a$ is assigned to agent $j$. In that case, it is unclear that removing item $a$ from $j$'s bundle decreases $i$'s ``envy'' towards $j$, since $i$ actually derives more utility when the item is allocated to $j$ than when it is allocated to $i$ herself. 
This issue becomes more complicated when both positive and negative externalities are allowed.

Another widely studied fairness concept under additive valuations is \emph{proportionality}, which requires each agent to receive at least $1/n$ of the value that she has for the set of all items, where $n$ denotes the number of agents. \citet{CFS17a} proposed a variant of proportionality for a more general \emph{public decision making} problem than allocation of indivisible items under externalities. They showed that their concept is guaranteed to be feasible under positive valuations. 
However, this guarantee ceases to hold when negative valuations are also allowed.

We summarize our contributions as follows.
First, we define the concepts of EF1 and EFX under externalities that still coincide with previous definitions for goods and chores when externalities do not exist.
Note that our new concepts work for both positive and negative externalities. 

Second, we show how to compute an EFX allocation between two agents in time $O(m \log m)$, where $m$ denotes the number of items, and how to compute an EF1 allocation between two agents in time linear in $m$.

Third, we show that the set of EFX allocations among three agents could be empty. Under binary values and a ``no-chore'' assumption, we show that an EF1 allocation always exists among three agents by proposing a new algorithm that computes such an allocation in polynomial time. 

Fourth, we propose a new fairness concept called \emph{general fair share (GFS)} based on proportionality. We present a polynomial-time algorithm that computes an allocation satisfying general fair share up to one item (GFS1) for the more general public decision making model where both positive and negative valuations are allowed. 

Finally, we present a taxonomy of fairness definitions including both existing and newly proposed concepts. 

\section{Related Work}
Fair allocation of indivisible items is an active topic of research in computer science and economics~\citep{BrTa96a,BCM15a,Thom15a}. For some recent overviews, we refer to the surveys of \citet{ABF+22a} and \citet{ALM+22a}.

For allocation problems under externalities, fairness concepts need to be carefully revisited and extended. \citet{Vele16a} proposed a natural adaptation of envy-freeness which requires that no agent prefers the allocation obtained by swapping her bundle with another agent. 
\citet{SPZ13a} considered both the envy-freeness concept of \citet{Vele16a} and proportionality in the context of cake-cutting. 
\citet{LZZ15a} studied truthful mechanisms in the setting of cake-cutting under externalities.
Since cake-cutting involves the allocation of a divisible resource, one can obtain existence results without relaxations even when externalities are present.

In our paper, we focus on allocation of indivisible items. \citet{SSG21a} presented an algorithm for computing an allocation that satisfies a relaxation of a concept called maximin share fairness, which can in turn be viewed as a relaxation of proportionality.
Note that both \citet{SSG21a} and \citet{SPZ13a} restricted their attention to settings with positive externalities, whereas we allow both positive and negative externalities.
\citet{LZZ15a} made the restrictive assumption that agents derive externalities that are percentages of other agents' values.
\citet{MPG21a} studied a special form of externalities in which an agent receives the same externality from an item regardless of which other agent receives the item.

A related line of work concerns house or residential allocation with externalities, where an agent's value for an allocation is influenced by other agents assigned to her neighborhood \citep{CLM18a,MaSi19a,EPTZ20a,AEGI21a,GBBM21a}.

\section{Model}

We consider a setting where a set of indivisible items $A = \{a_1, \dots, a_m\}$ are to be allocated among a set of agents $N = \{1, \dots, n\}$ under additive externalities. 

An allocation is denoted by $\pi = (\pi_1, \dots, \pi_n)$ where each $\pi_i \subseteq A$ is the bundle assigned to agent $i$ such that for any distinct $i,j \in N$, we have $\pi_i \cap \pi_j = \emptyset$. If $\bigcup_{i \in N} \pi_i = A$, then we call $\pi$ a \emph{complete} allocation of $A$. Unless specified otherwise, we only consider complete allocations. Let $\Pi$ denote the set of all allocations. For any item $a \in A$, let $\pi(a)$ denote the agent who receives item $a$ in allocation $\pi$.

Every agent $i \in N$ is associated with a valuation function $V_i: \Pi \rightarrow \mathbb{R}$, which assigns a real value to every allocation $\pi \in \Pi$. 
We assume that agents have additive valuations and externalities.
Under the additive preference domain, we have 
$V_i(\pi)$ $=$ $\sum_{a \in A}$ $V_i(\pi(a), a)$, where we abuse notation and let $V_i(j, a)$ represent the value that agent $i$ receives when item~$a$ is assigned to agent~$j$. Note that in problems without externalities, an agent receives the same value from an allocation as long as the agent receives the same bundle.

\section{EF1 and EFX under Externalities}

In this section, we consider how to generalize the definitions of EF1 and EFX to the setting of externalities. Note that we need to carefully design both definitions to ensure that they coincide with the previous definitions without externalities.

\citet{Vele16a} proposed a natural adaptation of envy-freeness which requires that no agent prefers the allocation obtained by swapping her bundle with another agent's bundle. Since this notion has become the standard of envy-freeness in the setting of externalities, we simply refer to it as envy-freeness.
In this work, we follow this idea of \emph{swapping bundles} to define EF1 and EFX.
Let $\pi^{i \leftrightarrow j}$ represent a new allocation in which only agents $i$ and $j$ swap their bundles in $\pi$ while other agents' bundles remain the same.

\begin{definition}[Envy-Freeness~\citep{Vele16a}]
\label{def:SEF}
An allocation $\pi$ is envy-free (EF) if there do not exist agents $i, j \in N$ such that 
$V_i(\pi^{i \leftrightarrow j}) > V_i(\pi)$.
\end{definition}

Recall that envy-freeness cannot be guaranteed in the indivisible domain even if there are two agents and one item, and the agents have no externalities. In view of this challenge, a natural recourse is to explore ``up to one item relaxations'' of fairness concepts. 
The intuition is that when an agent is envious, she would like to swap her bundle with another agent. The ``up to $k$'' relaxation ensures that such a swap is not desirable if at most $k$ items are removed from consideration. 
We next formalize an ``up to $k$ items relaxation'' of EF under externalities. 
      \begin{definition}[Envy-Freeness up to $k$ Items]
    \label{def:EFc_new}
   An allocation $\pi$ is envy-free up to $k$ items (EF$k$) if for every pair of agents $i,j\in N$, there exists a set of items $C \subseteq A$ and an allocation $\lambda$ such that the following conditions hold:
      \begin{enumerate}
        \item $|C|\leq k$;
        \item $\lambda_{\ell}=\pi_{\ell} \setminus C$ for all $\ell \in N$;
        \item  $V_i(\lambda)\geq V_i(\lambda^{i\leftrightarrow j})$.
      \end{enumerate}
  \end{definition}
  
In words, Definition~\ref{def:EFc_new} states that an allocation $\pi$ is EF$k$ if for each pair of agents $i$ and $j$, there exists a set of items $C$ of size at most $k$ such that for the new allocation $\lambda$ obtained by removing items in $C$ from each agent $\ell$'s bundle $\pi_{\ell}$ in $\pi$, agent $i$ would not like to swap her bundle $\lambda_{i}$ with agent~$j$'s bundle $\lambda_j$.
Note that if $k=1$ and there are no externalities, then Definition~\ref{def:EFc_new} coincides with the EF1 concept as formalized by \citet{Budi11a} for goods and by \citet{ACIW19a} for goods and chores.

  We next generalize EFX to the case of externalities. 
  A first attempt is to define the generalization so that if agent $i$ envies agent $j$, then removing any item from either of their bundles should eliminate the envy.
  However, this fails to capture the original idea of~\citet{CKM+19a} when externalities exist, as shown in Example~\ref{exmaple:EFX_motivation}.

\begin{example}
\label{exmaple:EFX_motivation}
Consider two agents $N = \{1, 2\}$ and three items $A = \{a, b, c\}$. The values of items and externalities are described in Table~\ref{table:SEFX}. 
For allocation $\pi = \{(1, ab), (2, c)\}$ in which agent $1$ receives items $a,b$ and agent $2$ receives item~$c$, agent $2$ has envy towards agent $1$:
  \[
  V_2(\pi) - V_2(\pi^{1 \leftrightarrow 2}) = (1+2+2) - (4 + 1 + 3) = -3.
  \]
  If we remove item $b$ from agent $1$'s bundle, then agent~$2$ envies agent $1$ even more. 
  That is, for allocation $\tilde{\pi}$ $= \{(1, a), (2, c)\}$,
  \[
  V_2(\tilde{\pi}) - V_2(\tilde{\pi}^{1 \leftrightarrow 2}) = (1+2) - (4+3) = -4.
  \]
\begin{table}[h]
\centering
\begin{tabular}{l|c|c|c}
 \hline
 & $a$ & $b$ & $c$
 \\
 \hline
 $1$ & $3, 1$ & $1, 2$ & $2, 1$ \\
  \hline
$2$ & $1, 4$ & $2, 1$ & $3, 2$ \\
 \hline
\end{tabular}
\caption{For each row $i \in \{1, 2\}$ and each entry $(x,y)$ in row~$i$, $x$ and $y$ denote the value that agent $i$ receives when the corresponding item is assigned to agent $1$ and $2$, respectively.}
\label{table:SEFX}
\end{table}
\end{example}

When we consider only goods (i.e., indivisible items with positive values) without externalities, if some agent $i$ envies another agent $j$, 
then removing any item from agent $j$'s bundle decreases the envy. However, this is no longer true when externalities exist.
As shown in Example~\ref{exmaple:EFX_motivation}, removing item~$b$ from agent~$1$'s bundle does not reduce the envy from agent~$2$ towards agent~$1$. Instead, it increases this envy.

We thus propose a more suitable generalization of EFX in Definition~\ref{def:EFX}. Intuitively, if agent $i$ envies agent $j$, then for any item $a$ such that removing $a$ from the bundle of agent $i$ or $j$ reduces the envy, $i$ should no longer envy $j$ after removing $a$. This idea coincides with the definition of EFX for goods and chores when there are no externalities \citep{ACIW19a}. Recall that in the definition by \citet{ACIW19a}, agent $i$'s envy towards agent $j$ can be eliminated by removing $i$'s least preferred good from $j$'s bundle as well as by removing $i$'s favorite chore (i.e., one yielding the least disutility) from $i$'s own bundle.

    \begin{definition}[Envy-Freeness up to Any Item]
    \label{def:EFX}
            An allocation $\pi$ is envy-free up to any item (EFX) if for all agents $i,j\in N$, if $i$ envies $j$, then for any item $a \in A$ and allocation $\lambda$ with the properties
            \begin{enumerate}
              \item $\lambda_{\ell}=\pi_{\ell}\setminus \{a\}$ for all $\ell\in N$ and
              \item $V_i(\lambda) - V_i(\lambda^{i \leftrightarrow j}) > V_i(\pi) - V_i(\pi^{i\leftrightarrow j}),$
            \end{enumerate}
          the following holds: 
            $$V_i(\lambda)\geq V_i(\lambda^{i\leftrightarrow j}).$$
  \end{definition}

We next explain Definition~\ref{def:EFX} in detail. $V_i(\pi) - V_i(\pi^{i\leftrightarrow j})$ represents the envy from agent $i$ towards agent $j$ with respect to allocation $\pi$. Because agent $i$ envies agent $j$, we have $V_i(\pi) < V_i(\pi^{i\leftrightarrow j})$, which implies $V_i(\pi)$ $-$ $V_i(\pi^{i\leftrightarrow j})$ $<$ $0$. Allocation $\lambda$ is obtained by removing some item $a$ from the bundle $\pi_{\ell}$ containing $a$.
(Note that if $a\not\in\pi_i\cup \pi_j$, then the envy of $i$ towards $j$ does not change upon removing $a$, so we may assume that $a\in\pi_i\cup \pi_j$.)
Similarly, $V_i(\lambda) - V_i(\lambda^{i \leftrightarrow j})$ represents the envy from $i$ towards $j$ with respect to the new allocation $\lambda$. Thus $V_i(\lambda) - V_i(\lambda^{i \leftrightarrow j}) > V_i(\pi) - V_i(\pi^{i\leftrightarrow j})$ means that removing item $a$ reduces the envy of $i$ towards $j$. Finally, $V_i(\lambda)\geq V_i(\lambda^{i\leftrightarrow j})$ requires that $i$ does not envy $j$ with respect to the new allocation $\lambda$.

Note that our new definition of EFX in Definition~\ref{def:EFX} still implies EF1 in Definition~\ref{def:EFc_new}. 
To see this, consider an EFX allocation $\pi$. 
For any pair of agents $i$ and $j$, if $i$ envies $j$ in allocation $\pi$, then because valuations and externalities are additive, there must exist an item $a \in A$ such that removing $a$ from either bundle $\pi_i$ or $\pi_j$ helps decrease the envy of $i$ towards $j$.
Since $\pi$ is EFX, removing $a$ must eliminate $i$'s envy towards $j$.
Thus the allocation $\pi$ is also EF1. We next show an example of an EF1 allocation that is not EFX.

\begin{example}
\label{example:EF1_not_EFX}
Consider the instance in Example~\ref{exmaple:EFX_motivation}. Allocation $\pi' = \{(1, bc), (2, a)\}$ is EF1 but not EFX for agent~$1$.
To see this, first note that agent $1$ envies agent $2$:
  \[
  V_1(\pi') - V_1(\pi'^{1 \leftrightarrow 2}) = (1+2+1) - (3+2+1) = -2.
  \]
If we remove item $a$ from agent $2$'s bundle $\pi_2'$, then agent $1$ does not envy agent~$2$.
  That is, for allocation $\hat{\pi}$ $= \{(1, bc), (2, \emptyset)\}$,
  \[
  V_1(\hat{\pi}) - V_1(\hat{\pi}^{1 \leftrightarrow 2}) = (1+2) - (2+1) = 0.
  \]
On the other hand, if we remove item $b$ from agent $1$'s bundle $\pi_1'$, which helps decrease agent~$1$'s envy, then agent $1$ still envies agent $2$.  That is, for allocation $\overline{\pi}$ $= \{(1, c), (2, a)\}$,
  \[
  V_1(\overline{\pi}) - V_1(\overline{\pi}^{1 \leftrightarrow 2}) = (1+2) - (3+1) = -1.
  \]
\end{example}

\section{Two Agents}
\label{sec:two-agents}

In this section, we prove the existence of EFX allocations between two agents by mapping onto a simplified problem where agents have ``symmetric valuations'' and an EFX allocation can be constructed in polynomial time. 
In Lemma~\ref{lemma:2same}, we show how to construct an EFX allocation between two agents when valuations are symmetric. Based on this result, we then prove the existence of EFX allocations between two agents in Theorem~\ref{theo:EFX:two}. 
We further show that an EF1 allocation between two agents can be computed in linear time.

We say that agents' valuations are \emph{symmetric} if for each pair $i,j\in N$ and each item $a \in A$, the following holds: $V_i(i,a)=V_j(j,a)$ and $V_i(j,a)=V_j(i,a)$. 

\begin{lemma}
\label{lemma:2same}
An EFX allocation always exists for two agents with symmetric valuations and can be computed in time $O(m \log m)$.
\end{lemma}

\begin{proof}
Consider two agents with symmetric valuations. For an item $a$, let $\Delta_{12}(a) = V_1(1, a) - V_1(2, a)$ be the difference between the value $V_1(1, a)$ that agent $1$ receives when $a$ is assigned to agent $1$ and the value $V_1(2, a)$ that she receives when $a$ is assigned to agent $2$.
Since valuations are symmetric, $V_1(1,a)-V_1(2,a) = V_2(2,a) - V_2(1,a)$.

Create an allocation $\tilde{\pi}$ as follows. We iteratively allocate each item in decreasing order of $|\Delta_{12}(\cdot)|$. At each step, there are two possible bundles for the item, leading to two different allocations. 
We choose one that the agent with the smaller current total value weakly prefers from these two allocations. That is, if $\Delta_{12}(a) \geq 0$, then the item is assigned to the agent with the smaller current total value; otherwise the item is assigned to the other agent. Break ties arbitrarily.

We next prove that allocation $\tilde{\pi}$ is EFX for both agents. Suppose we allocate all items in the order $a_1, a_2, \ldots, a_m$.
For the base case, assigning $a_1$ to either agent is EFX. For the induction, assume that a partial allocation of items $a_1$, $\ldots$ , $a_{k}$ is EFX. 
Since the agents have symmetric valuations, at most one agent can be envious. Without loss of generality, assume agent $1$ has at most the same value as agent~$2$ and 
the algorithm allocates $a_{k+1}$ according to agent $1$'s preference. Then agent $1$'s envy towards agent $2$ weakly decreases and the allocation is still EFX for agent $1$. If agent $2$ becomes envious, then removing item $a_{k+1}$ will eliminate the envy. For any item $a_j$ allocated to agent $1$ with $\Delta_{12}(a_j) > 0$, we have $\Delta_{12}(a_j) \geq \Delta_{12}(a_{k+1})$ and removing any such item will eliminate the envy from agent $2$ as well; a similar argument holds for any item $a_j$ allocated to agent~$2$ with $\Delta_{12}(a_j) < 0$. Thus the allocation remains EFX for both agents. 

We can sort all items in decreasing order of $|\Delta_{12}(\cdot)|$ in time $O(m \log m)$ and thus we can compute an EFX allocation between two agents with symmetric valuations in polynomial time. This completes the proof of Lemma~\ref{lemma:2same}.
\end{proof}

Based on Lemma~\ref{lemma:2same}, we prove the existence of EFX allocation between two agents in Theorem~\ref{theo:EFX:two}.
\begin{theorem}
\label{theo:EFX:two}
There always exists an EFX allocation between two agents which can be computed in time $O(m \log m)$.
\end{theorem}

\begin{proof}
First create a dummy agent $1'$ of $1$. Both agent $1$ and agent $1'$ treat each other as agent $2$ and they have a symmetric valuation function such that for any item $a \in A$, we have $V_1(1, a) = V_{1'}(1', a)$ and $V_1(1', a) = V_{1'}(1, a) = V_1(2, a)$.
That is, if item $a$ is assigned to $1'$, then agent $1'$ receives the value $V_1(1, a)$ and agent $1$ receives the value $V_1(2, a)$ as if item $a$ is assigned to $2$ from the perspective of agent $1$. 

Compute an EFX allocation $\tilde{\pi}$ between agents $1$ and $1'$ via the algorithm in the proof of Lemma~\ref{lemma:2same}.
Allocation $\tilde{\pi}$ divides all items $A$ into two bundles; let agent $2$ first choose the bundle she prefers and leave the remaining bundle to agent $1$. Since agent $2$ chooses first, she does not envy agent~$1$. We showed that $\tilde{\pi}$ is EFX between $1$ and $1'$ in Lemma~\ref{lemma:2same}, so it is EFX no matter which bundle agent $1$ receives. This completes the proof of Theorem~\ref{theo:EFX:two}.
\end{proof}

We remark that constructing an EF1 allocation is easier and can be done in linear time, because we do not need to sort all items based on $|\Delta(\cdot)|$. The detailed proof is provided in \Cref{app:proof-cor-EF1-two}.

\begin{corollary}
\label{cor:EF1-two}
There always exists an EF1 allocation between two agents which can be computed in time $O(m)$.
\end{corollary}

\section{Three Agents}
\label{sec:three-agents}

In this section, we consider EF1 and EFX allocations among three agents. Note that several real-world problems involve a limited number of agents (e.g., divorce settlement and inheritance division).
We first show that, in contrast to the positive results of EFX allocations between two agents with externalities (Section~\ref{sec:two-agents}) and among three agents \emph{without} externalities \citep{CGM20a}, there may not exist an EFX allocation for three agents with externalities.
We then prove that an EF1 allocation always exists among three agents under binary values and a ``no-chore'' assumption by proposing a polynomial-time algorithm for this case.

\begin{theorem}
\label{theo:non:EFX:three}
The set of EFX allocations could be empty when there are three agents.
\end{theorem}

\begin{proof}
We prove Theorem~\ref{theo:non:EFX:three} through the following counterexample. Consider three agents $N = \{1, 2, 3\}$ and seven items $A = \{a_1, a_2, a_3, a_4, a_5, a_6, g\}$. The values of items and externalities are described in Table~\ref{table:EFX}. 

\begin{table}[h]
\centering
\begin{tabular}{l|c|c}
 \hline
 & $a_k$ & $g$ 
 \\
 \hline
 1 & $21, 16, 16$ & $17, 16, 16$ \\
  \hline
 2 & $16, 21, 16$ & $16, 24, 0$ \\
  \hline
 3 & $16, 16, 21$ & $16, 0, 24$ \\
 \hline
\end{tabular}
\caption{Values for items and externalities. For each row $i \in \{1, 2, 3\}$ and each entry $(x,y,z)$ in row~$i$,  $x$, $y$, and $z$ denote the value that agent $i$ receives when the corresponding item is assigned to agent $1$, $2$, and $3$, respectively.}
\label{table:EFX}
\end{table}

\noindent
First consider the case where item $g$ is assigned to agent $1$. 
\begin{itemize}
\item If at most one item from $\{a_1, \dots, a_6\}$ is assigned to agent $1$, then either agent $2$ or $3$ receives at least three items from this set. Suppose agent $1$ receives $\{a_4, g\}$ (or just $\{g\}$) and agent $2$ receives $\{a_1, a_2, a_3\}$ (or more). Then agent $1$ has envy towards agent $2$, and the envy remains when one item $a_1$ is removed from agent $2$'s bundle.
\item If at least two items from $\{a_1, \dots, a_6\}$ are assigned to agent $1$, then either agent $2$ or $3$ receives at most two items from this set. 
Suppose agent $1$ receives $\{a_1, a_2, g\}$
(or more) and agent $2$ receives $\{a_3, a_4\}$ (or less). Then agent $2$ has envy towards agent $1$, and the envy remains when one item $a_1$ is removed from agent $1$'s bundle.
\end{itemize}
\noindent
Consider now the case where item $g$ is assigned to agent $2$ or $3$.
By symmetry, assume that it is assigned to agent $2$.
We divide into two cases.

\begin{itemize}
\item No item from $\{a_1, \dots, a_6\}$ is assigned to agent $2$:
\begin{itemize}
\item If the six items are not evenly divided between agent $1$ and $3$, say agent $1$ receives $\{a_1, a_2\}$ (or less) and agent $3$ receives $\{a_3, a_4, a_5, a_6\}$ (or more), then agent $1$ envies agent $3$ and the envy remains even if one item $a_3$ is removed from agent $3$'s bundle. A similar argument works when agent $3$ receives $\{a_1, a_2\}$ (or less) and agent $1$ receives $\{a_3, a_4, a_5, a_6\}$ (or more).
\item If the six items are evenly divided between agent $1$ and $3$, say agent $1$ receives $\{a_1, a_2, a_3\}$, then agent $2$ envies agent $1$ and the envy remains even if one item $a_1$ is removed from agent $1$'s bundle. 
\end{itemize}
\item Agent $2$ receives at least one item from $\{a_1, \dots, a_6\}$:
\begin{itemize}
\item If agent $3$ receives no more than three items from $\{a_1, \dots, a_6\}$, say agent $2$ receives $\{a_1, g\}$ (or more) and agent $3$ receives $\{a_2, a_3, a_4\}$ (or less), then agent~$3$ has envy toward agent $2$ even if one item $a_1$ is removed from agent $2$'s bundle. 
\item If agent $3$ receives more than three items from $\{a_1, \dots, a_6\}$, then agent $1$ receives at most one item from $\{a_1, \dots, a_6\}$ and agent $1$ has envy towards agent~$3$ even if one item $a_i$ is removed from the assignment of agent $3$.
\end{itemize}
\end{itemize}
For all possible cases, there does not exist an EFX allocation. This completes the proof.
\end{proof}

We next consider EF1 allocations and restrict our attention to the case where the three agents have binary valuations, i.e., each $V_i(j,a)$ is either $0$ or $1$. Note that binary valuations without externalities have previously been considered in fair division both for the desirable normative properties that they allow and for their ease of elicitation~\citep{AAGW15b,DaSc15a,BoLe15a,BKV18a,FSVX19a,HPPS20a,SuTe22a}. 

 Even under the setting of binary valuations, it still appears to be challenging to check if a given instance admits an EF1 allocation and to determine whether there always exists an EF1 allocation for any instance. To make the question tractable, we additionally impose one more assumption in Definition~\ref{assumption:no-chores}: for each item, all agents prefer to own the item rather than to have it assigned to others. Note that the example in Theorem~\ref{theo:non:EFX:three} satisfies this assumption.
 
\begin{definition}[No-Chore Assumption]
\label{assumption:no-chores}
For any item $a$ and any pair of agents $i, j$, we have $V_{i}(i,a) \geq V_{i}(j,a)$.
\end{definition}
 
 We prove that an EF1 allocation always exists in this setting and can be computed in polynomial time. Our method may be useful for further exploration of more general settings, e.g., whether there exists an EF1 allocation among multiple agents under more general preference domains.

\begin{theorem}
\label{theorem:EF1}
 For three agents under No-Chore Assumption and binary valuations, there always exists an EF1 allocation which can be computed in polynomial time.
\end{theorem}

We give a high-level description of our algorithm here and present a detailed proof of Theorem~\ref{theorem:EF1} in Appendix~\ref{app:three-no-chore}. 
The key idea is that given any instance, we iteratively apply some reduction rules that assign one item, one pair, or three items to some agents in an EF or EF1 manner. We show that any instance can be reduced to a certain number of cases where each case consists of a small number of items (no more than $12$). We then wrote a program to verify that for each case there always exists an EF1 allocation by exhaustive search.

For the sake of illustration, we next describe two simple reduction rules. Given an item $a \in A$, we can create a $3$-by-$3$ matrix to represent each agent $i$'s valuation function $V_i(\cdot, a)$, where the $i$th row corresponds to the values that agent $i$ receives when item $a$ is assigned to each agent.
\begin{itemize}
\item Assign an item $a$ to some agent $i$ if it does not generate envy from any agent towards $i$. For instance, if we have an item with the following matrix, then we can assign it to agent $1$ without generating envy from any other agent.
$$
\begin{bmatrix}
1 & 0 & 0 \\
1 & 1 & 0 \\
1 & 0 & 1 
\end{bmatrix}
$$
\item Suppose that item $a$ will not generate envy from agent $i$ if it is assigned to any other agent. 
Then we can leave this item aside until we cannot apply any other reduction rules and then consider
assigning item $a$ to the other two agents in an EF1 manner. 
For instance, if we have an item with the following matrix, then assigning it to either agent $2$ or agent $3$ does not generate envy from agent $1$.
$$
\begin{bmatrix}
1 & 1 & 1 \\
0 & 1 & 0 \\
0 & 0 & 1 
\end{bmatrix}
$$
\end{itemize}

Note that we have only taken an initial step towards a complete understanding of EF1 allocations under externalities.
We conjecture that an EF1 allocation always exists for three agents under binary valuations even without the No-Chore Assumption. For larger numbers of agents $n$, we may need to relax EF1 to EF$k$ where $k$ is a function of $n$.

\section{GFS and Public Decision Making}
\label{sec:public-decision-GFS}

In this section, we propose a new fairness concept based on proportionality that we call \emph{general fair share (GFS)}.
This concept works even for the ``public decision making'' setting \citep{CFS17a}, which generalizes fair division of indivisible items under externalities.
We show that there always exists an allocation satisfying general fair share up to one item (GFS1) even when the valuations can be positive or negative, and such an allocation can be computed in polynomial time via a variant of round robin. 
We also discuss how GFS1 is superior to an existing proportionality concept in public decision making.

\subsection{Public Decision Making}
\label{sec:public}
An instance $I^{P}$ of public decision making consists of a set of agents $N$ and a set of issues $A$. 
Each issue $a \in A$ is associated with a set of choices $a^T$, exactly one of which needs to be selected. For each choice $a^t \in a^T$ of issue $a$, each agent~$i$ derives a value $V_i(a^t)$, where we reuse the notation $V_i$ in a slightly different way than in fair allocation under externalities.
An allocation $\pi$ of instance $I^P$ is a set of choices for all issues; let $\pi(a)$ denote the choice made for issue $a$. The value that agent $i$ receives from allocation $\pi$ is $V_i(\pi) = \sum_{a\in A} V_i(\pi(a))$.

A fair allocation problem with additive externalities can be reduced to an equivalent public decision making problem as follows: Each item $a$ is viewed as an issue and associated with exactly $n$ choices, where each choice corresponds to an agent to whom the item could be given. 
For public decision making, the number of choices is flexible, whereas for fair allocation with externalities, the number of choices is equal to the number of agents $n$. 

\citet{CFS17a} proposed the following concept for the public decision making problem, which requires that each agent $i$ should receive at least $1/n$ of the maximum value she can get from all of the issues.\footnote{
In their paper, this concept is simply called proportionality, but we refer to it as PROP-Max to distinguish it from another variant of proportionality that we will discuss later.}
For each issue~$a$, let $V_i^{max}(a) = \max_{a^t\in a^T} V_i(a^t)$

\begin{definition}[PROP-Max]
\label{def:PROP_max_appendix}
Given an allocation $\pi$, the Proportional-Max share of agent $i$ (PROP-Max$_i$) is defined as
\[
\text{PROP-Max}_i = \frac{1}{n} \sum_{a\in A} V_i^{max}(a).
\]
An allocation $\pi$ satisfies Proportionality-Max (PROP-Max) if $V_i(\pi) \geq$ PROP-Max$_i$ holds for all $i \in N$.
\end{definition}

Conitzer et al.~also introduced an ``up to one'' relaxation of PROP-Max.

\begin{definition}[PROP-Max up to One Issue]
\label{def:PROP_max_1}
An allocation $\pi$ satisfies Proportionality-Max up to one issue (PROP-Max-1) if for all $i \in N$, there exists $a \in A$ such that
\[
V_i(\pi) - V_i(\pi(a)) + V_i^{max}(a) \geq \text{PROP-Max}_i. 
\]
\end{definition}

In other words, an allocation $\pi$ satisfies PROP-Max-1 if for each agent $i$, there exists an item $a$ such that changing the assignment of $a$ from $\pi(a)$ to agent $i$'s best assignment yielding $V_i^{max}(a)$ ensures that the value that agent $i$ receives is at least her PROP-Max$_i$.

Conitzer et al.~showed that when all valuations are positive, there always exists a PROP-Max-1 allocation.
However, our next proposition shows that a PROP-Max-1 may not exist if negative valuations are allowed.

\begin{proposition}
\label{prop:non:max1}
There may not exist a PROP-Max-1 allocation when negative valuations are allowed, even if there are only two agents.
\end{proposition}

\begin{proof}
We show this negative result for the more restricted setting of fair allocation with externalities. 

Consider $N=\{1, 2\}$ and $A= \{a_1, a_2, a_3\}$. Suppose that for distinct $i, j\in N$ and each item $a\in A$, we have $V_i(i, a) = 0$ and $V_i(j, a) = -100$. One agent (say, $1$) must receive at least two items, and the value of agent $2$ is $-200$. However, agent $2$'s maximum value is $0$, but it is not possible to attain this by reassigning one item. 
\end{proof}

\subsection{GFS and GFS1 Concepts}

The incompatibility of PROP-Max-1 with negative valuations in \Cref{prop:non:max1} motivates us to propose a new fairness concept called \emph{general fair share (GFS)}.
For agent~$i$ and issue~$a$, let $V_i^{min}(a) = \min_{a^t\in a^T} V_i(a^t)$.

\begin{definition}[General Fair Share]
\label{def:GFS}
The general fair share of agent $i$ (GFS$_i$) is defined as
\begin{align*}
\text{GFS}_i 
&=  \frac{1}{n} \sum_{a\in A}  V_i^{max}(a) + \frac{n-1}{n} \sum_{a\in A} V_i^{min}(a) \\
&= \sum_{a\in A} V_i^{min}(a) + \frac{1}{n} \sum_{a\in A} (V_i^{max}(a) - V_i^{min}(a)).
\end{align*}
An allocation $\pi$ satisfies general fair share (GFS) if $V_i(\pi) \geq  \text{GFS}_i$ for all $i \in N$.
\end{definition}

We next illustrate the intuition of GFS. Consider a GFS allocation $\pi$. For any agent $i$, the improvement that $\pi$ offers upon agent $i$'s worst allocation is at least $\frac{1}{n} \sum_{a\in A} (V_i^{max}(a) - V_i^{min}(a))$, i.e., $1/n$ of agent~$i$'s largest possible improvement---the improvement required by GFS is shown in the colored regions of Figure~\ref{fig:Prop-max-min}.
That is, if we subtract $\sum_{a\in A} V^{min}_i(a)$ from $V_i(\pi)$, then we have
\begin{align*}
V_i(\pi) - \sum_{a\in A}  V^{min}_i(a)  &\geq \text{GFS}_i - \sum_{a\in A} V^{min}_i(a) 
\\
& = \frac{1}{n} \sum_{a\in A} (V_i^{max}(a) - V_i^{min}(a)).
\end{align*}

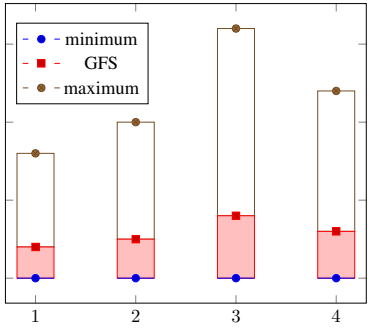
\begin{figure}[tb]
  \centering
\begin{tikzpicture}[scale=0.70]
\begin{axis}[stack plots=y,/tikz/ybar, bar width=20pt, label style={font=\large}, legend style={at={(0.03,0.8)},anchor=west},
symbolic x coords={$1$, $2$, $3$, $4$},
xtick=data, yticklabels={,,} 
]
  \addplot+[ybar] plot coordinates {($1$,0) ($2$,0) 
  ($3$, 0) ($4$,0) };
\addplot+[ybar, fill= pink] plot coordinates {($1$,0.2) ($2$,0.25) 
  ($3$, 0.4) ($4$,0.3) };
\addplot+[ybar] plot coordinates {($1$,0.6) ($2$,0.75)  ($3$, 1.2) ($4$,0.9) };
  \legend{\strut minimum, \strut GFS, \strut maximum
  }
\end{axis}
\end{tikzpicture}
\caption{Illustration of GFS with four agents. 
Bottom lines and top lines denote the minimum and the maximum values each agent can receive among all possible allocations.
Middle lines denote the GFS for each agent, and the colored region in each bar equals one-fourth of the difference between the maximum value and the minimum value of that agent.
}
\label{fig:Prop-max-min}
\end{figure}

Similarly to PROP-Max, GFS is too strong to guarantee corresponding allocations, so we relax it in the same manner as PROP-Max-1. 
We refer to this concept as \emph{general fair share up to one item (GFS1)}.

\begin{definition}[General Fair Share up to One Item]
\label{def:GFS1}
An allocation $\pi$ satisfies general fair share up to one item (GFS1) if for all $i \in N$, there exists $a \in A$ such that
\[
V_i(\pi) - V_i(\pi(a)) + V_i^{max}(a) \geq \text{GFS}_i. 
\]
\end{definition}
In other words, an allocation $\pi$ satisfies GFS1 if for each agent $i$, there exists an item $a$ such that changing the assignment of $a$ from $\pi(a)$ to agent $i$'s best assignment yielding $V_i^{max}(a)$ ensures that the value that agent $i$ receives is at least her general fair share GFS$_i$.

With positive valuations, our GFS/GFS1 concepts are stronger than PROP-Max/PROP-Max-1.

\begin{proposition}
\label{prop:GFS_Max}
For public decision making with positive valuations, GFS implies PROP-Max, and GFS1 implies PROP-Max-1.
\end{proposition}

\begin{proof}
Given a GFS allocation $\pi$, for any agent $i$ we have 
\begin{align*}
V_i(\pi) - \frac{1}{n}\sum_{a\in A} & V^{max}_i(a) \geq \frac{n-1}{n} V_i^{min}(a) \geq 0,
\end{align*}
where we use the assumption of positive valuations for the latter inequality.
Thus, GFS implies PROP-Max. The proof that GFS1 implies PROP-Max-1 is similar.
\end{proof}

We next show that a GFS1 allocation always exists.
Combined with \Cref{prop:non:max1,prop:GFS_Max}, this means that GFS1 is a more suitable concept in public decision making than PROP-Max-1, both when valuations are only positive and when negative valuations are allowed.

\subsection{Max-Min Round Robin}
In this subsection, we present a polynomial-time algorithm ``Max-Min Round Robin'' that computes a GFS1 allocation for public decision making. 

We give here a brief description of Max-Min Round Robin. For agent $i\in N$ and issue $a \in A$, let $\beta_i(a)$ $=$ $V_i^{max}(a)$ $-$ $V_i^{min}(a)$ denote the difference between the maximum and minimum value that agent $i$ can receive from issue $a$.
The Max-Min Round Robin algorithm works as follows: First, fix a round robin sequence of agents. Then, for each agent $i$'s turn, let $i$ determine the choice of some issue $a$ in her favor such that $\beta_i(a)$ is the largest among all remaining issues. Repeat this procedure until there are no issues left. The details are described in Algorithm~\ref{alg:MMRR}.
We reiterate that this algorithm works for both positive and negative valuations.

  \begin{algorithm}[h!]
\begin{algorithmic}[scale=1]
        \REQUIRE an instance of public decision making
        \ENSURE a GFS1 allocation
        \end{algorithmic}
       \begin{algorithmic}[1]
      \caption{Max-Min Round Robin}
      \label{alg:MMRR}
      \STATE Let $\beta_i(a) = V_i^{max}(a) - V_i^{min}(a), \forall i\in N, a\in A$.
      \STATE Fix a round robin sequence of agents, say $1, 2, \dots, n$.
      \STATE $\pi \leftarrow \emptyset, j \leftarrow 1$
      \WHILE{$A \neq \emptyset$}
      \STATE For agent $j$'s turn, find an issue $a$ and a choice $a^t$ such that 
      \begin{itemize}
      \item $\forall a' \in A, \beta_j(a) \geq \beta_j(a')$
      \item $V_j^{max}(a) =  V_j(a^t)$
      \end{itemize}
      \STATE $\pi \leftarrow \pi \cup \{a^t\}$ \COMMENT{Choose $a^t$ for issue $a$}
      \STATE $A \leftarrow A \setminus \{a\}$ \COMMENT{Remove $a$ from $A$}
      \STATE $j \leftarrow (j$ mod $n) + 1$ \COMMENT{Move on to the next agent}
      \ENDWHILE
      \RETURN an allocation $\pi$
       \end{algorithmic}
       \end{algorithm}

\begin{theorem}
\label{theo:max-min-RR}
Max-Min Round Robin returns a GFS1 allocation for public decision making in polynomial time.
\end{theorem}

\begin{proof}
Given an instance of public decision making, let $\pi$ denote the outcome yielded by Algorithm~\ref{alg:MMRR}. It is clear that the algorithm runs in polynomial time, since the algorithm traverses all issues once.
We next prove that, for any agent $i$, the value agent $i$ receives from $\pi$ satisfies GFS1.

Suppose there are a total of $m = p \cdot n + q$ issues, where $n$ denotes the number of agents and $1 \le q \le n$. Then Algorithm~\ref{alg:MMRR} terminates in $p+1$ rounds. Rank all issues $a \in A$ in descending order of $\beta_i(a) = V_i^{max}(a) - V_i^{min}(a)$, and divide the issues into $p+1$ groups as follows:
\[
a_1, a_2, \dots, a_n \mid a_{n+1}, \dots, a_{2n} \mid \dots \mid a_{pn+1}, \dots, a_{pn+q}
\]
During the execution of Algorithm~\ref{alg:MMRR}, agent $i$ determines some issue from $\{a_1, \dots, a_n\}$ in her favor for her first turn, then determines some issue from $\{a_{1}, \dots, a_{2n}\}$ in her favor for her second turn, and so on.
The minimum value that agent $i$ receives from allocation $\pi$ is 
\[
\min V_i(\pi) = \sum_{a \in A'} V_i^{max}(a) + \sum_{a \in A \setminus A'} V_i^{min}(a) 
\]
where $A' = \{a_n, a_{2n}, \dots, a_{pn}\}$. In other words, agent $i$ receives the minimum value $\min V_i(\pi)$ from allocation $\pi$ when agent $i$ receives $V_i^{max}(a)$ for issues from $\{a_n$, $a_{2n}$, $\dots$, $a_{pn}\}$ and receives $V_i^{min}(a)$ for all other issues.

For allocation $\pi$, we have that 
\begin{align*}
& V_i(\pi) - \sum_{a\in A} V_i^{min}(a) 
\\
&\geq  \min V_i(\pi) - \sum_{a\in A} V_i^{min}(a)
\\
&=  \sum_{a \in A'} V_i^{max}(a) + \sum_{a \in A \setminus A'} V_i^{min}(a) - \sum_{a\in A} V_i^{min}(a)
\\
&=  \sum_{a \in A'} V_i^{max}(a) - \sum_{a \in A'} V_i^{min}(a) = \sum_{k=1}^p \beta(a_{kn}).
\end{align*}

Suppose that during the process of Algorithm~\ref{alg:MMRR}, for agent $i$'s turns, she determines issues $b_1, b_2, \dots ,b_p$ in her favor which are ranked in descending order of $\beta_i(\cdot)$. In other words, agent $i$ receives $V_i^{max}(b)$ for each issue $b$ from $b_1, b_2, \dots ,b_p$.

Consider the first issue $x$ from $a_1, a_2, \dots, a_m$ that differs from $b_1, b_2, \dots ,b_p$. 
Then we create a new allocation $\pi'$ in which only one choice differs from $\pi$: agent $i$ receives $V^{max}_i(x)$ from issue $x$ instead of $V_i(\pi(x))$, while all other choices remain the same as $\pi$. For instance, if $x = a_1$,
then agent $i$ receives $V^{max}_i(a_1)$ instead of $V_i(\pi(a_1))$ in allocation $\pi'$; if $x = a_2$, then agent $i$ has received $V^{max}_i(a_1)$ in allocation $\pi$ and now she receives $V^{max}_i(a_2)$ from issue $a_2$ instead of $V_i(\pi(a_2))$ in allocation $\pi'$, and so on.

Thus, in allocation $\pi'$, agent $i$ can determine at least two issues from $a_1$, $a_2$, $\dots$, $a_n$ in her favor where one is $a_1$ and the other one is weakly better than $a_n$. While for the worst case of allocation $\pi$, agent $i$ only determines $a_n$ from $a_1$, $a_2$, $\dots$, $a_n$ in her favor. 
Then we have
\begin{equation*}
V_i(\pi') \geq \min V_i(\pi) + \beta_i(a_1)
\end{equation*}
which implies that 
\begin{equation}
\label{equ:1}
V_i(\pi') - \sum_{a\in A} V_i^{min}(a) \geq \beta_i(a_1) + \sum_{k=1}^p \beta_i(a_{kn}).
\end{equation}

Recall that $\beta_i(a_y) \geq \beta_i(a_z)$ for any $y < z$ with $y, z \leq m$.
For issue $a_1$ and $c \in [1, n]$, we have $\beta_i(a_1) \geq \beta_i(a_c)$, which implies that 
\begin{equation}
\label{equ:2}
\beta_i(a_1) \geq \frac{1}{n} \sum_{c=1}^n \beta_i(a_c).
\end{equation}

Assume that $V_i^{max}(a_{pn+c}) = V_i^{min}(a_{pn+c}) = 0$ for $c \in [q+1, n]$.
For issue $a_{kn}$ with $1 \leq k \leq p$ and $c \in [1, n]$ we have 
$\beta_i(a_{kn}) \geq \beta_i(a_{kn+c})$, which implies that 
\begin{equation}
\label{equ:3}
\beta_i(a_{kn}) \geq \frac{1}{n} \sum_{c=1}^n \beta_i(a_{kn+c}).
\end{equation}

From inequalities (1), (2) and (3), we have 
\begin{align*}
& V_i(\pi') - \sum_{a\in A} V_i^{min}(a) 
\\
&\geq  \beta_i(a_1) + \sum_{k=1}^p \beta_i(a_{kn})
\\
&\geq  \frac{1}{n}\sum_{c=1}^n \beta_i(a_c) + \frac{1}{n}\sum_{k=1}^p \sum_{c=1}^n \beta_i(a_{kn+c})
\\
&=  \frac{1}{n}\sum_{a\in A} \beta_i(a)= \frac{1}{n} \sum_{a\in A} ( V_i^{max}(a) - V_i^{min}(a) ).
\end{align*}

Thus the new allocation $\pi'$ satisfies GFS for each agent~$i$. Since we obtain allocation $\pi'$ from $\pi$ by modifying one issue~$x$, the allocation $\pi$ satisfies GFS1, completing the proof.
\end{proof}

\section{Taxonomy of Fairness Concepts}
\label{sec:fairness}
In this section, we present a taxonomy of fairness concepts for fair allocation with externalities including existing and newly proposed ones (Figure~\ref{fig:relation}). 

Besides PROP-Max, another extension of proportionality is PROP-Ave, proposed by \citet{SSG21a}.\footnote{These authors called the notion \emph{average-share}.}
Maximin Share (MMS) is a relaxation of proportionality for fair division of indivisible items, introduced by \citet{Budi11a}. 
\citet{SSG21a} proposed Extended Maximin Share (EMMS) which generalizes MMS to the case of externalities.

For fair division without externalities, EF implies proportionality. However, we show that EF implies neither PROP-Max nor PROP-Ave when externalities exist. We propose a new notion $k$-Partial-Proportionality ($k$-P-PROP) that connects both EF and PROP-Ave. The intuition is that for any subset of agents $N' \subseteq N$ with $|N'| \leq k$, each agent $i\in N'$ should receive at least $1 / |N'|$ of the total value
she can receive from all items assigned to the group $N'$. 

Note that \citet{ABC+18a} considered a general framework called $\mathcal{H}$-HG-PROP for defining fairness concepts when allocating indivisible items in the presence of a social graph. If there are no externalities and $\mathcal{H}$ is the set of hypergraph consisting of all subsets of size at most $k$ of the agents, then $k$-P-PROP is equivalent to $\mathcal{H}$-HG-PROP. 

\begin{figure}[tb]
    \begin{center}

\begin{tikzpicture}[scale=0.9]
  \tikzstyle{onlytext}=[]

  \node[onlytext] (EF) at (0, 0) {EF};
  \node[onlytext] (EFX) at (-1.25, -1.25) {EFX};
  \node[onlytext] (EF1) at (-1.25, -2.5) {EF1};
  \node[onlytext] (EFk) at (-1.25, -3.75) {EF$k$};  

  \node[onlytext] (2-P-PROP) at (1.25, -1.25) {2-P-PROP};

    \node[onlytext] (k-P-PROP) at (2.5, 0) {$n$-P-PROP};
    \node[onlytext] (PROP-Ave) at (3.75, -1.25) {PROP-Ave}; 
    \node[onlytext] (GFS) at (3.75, -2.5) {GFS};
    \node[onlytext] (GFS-1) at (6, -2.5) {GFS1};
    \node[onlytext] (PROP-Max) at (3.75, -3.75) {PROP-Max};

    \node[onlytext] (EMMS) at (6, -1.25) {EMMS};

\draw[->, line width=1pt] (EF) -- (EFX) ;
\draw[->, line width=1pt] (EFX) -- (EF1) ;
\draw[->, line width=1pt] (EF1) -- (EFk) ;

\draw[->, line width=1pt] (EF) -- (2-P-PROP) ;

\draw[->, line width=1pt] (k-P-PROP) -- (2-P-PROP) ;
\draw[->, line width=1pt] (k-P-PROP) -- (PROP-Ave) ;
\draw[->, line width=1pt] (PROP-Ave) -- (GFS) ;
\draw[->, line width=1pt] (GFS) -- (PROP-Max) node[midway, right] {positive};
\draw[->, line width=1pt] (GFS) -- (GFS-1);
\draw[->, line width=1pt] (PROP-Ave) -- (EMMS) ;
\end{tikzpicture}
\end{center}
\caption{\label{fig:relation} Relationships among fairness concepts for fair allocation with externalities. An arrow from A to B denotes that A implies B. EF, EMMS, PROP-Ave, and PROP-Max under externalities were proposed in previous work. 
}
\end{figure}
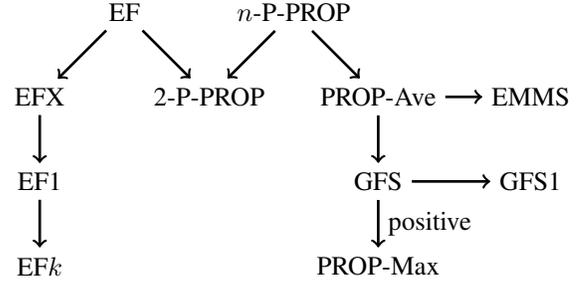

\subsection{Proportionality}

As we discussed in \Cref{sec:public-decision-GFS}, \citet{CFS17a} proposed the concept Proportionality-Max (PROP-Max) for the setting of public decision making.
We now state it in the context of fair allocation with externalities.

\begin{definition}[PROP-Max]
\label{def:PROP_max}
Given an allocation $\pi$, the Proportional-Max share of agent $i$ (PROP-Max$_i$) is
\[
\text{PROP-Max}_i = \frac{1}{n} \sum_{a\in A} V_i^{max}(a)
\]
where $V_i^{max}(a) = \max_{j \in N} V_i(j, a)$ denotes the maximum value agent $i$ can derive from item $a$.
An allocation $\pi$ satisfies Proportionality-Max (PROP-Max) if $V_i(\pi) \geq$ PROP-Max$_i$ holds for all $i \in N$.
\end{definition}

\begin{proposition}
\label{prop:EF&PROP-Max}
EF does not imply PROP-Max for two agents even when there are no externalities.
\end{proposition}

\begin{proof}
Consider two items $A = \{a_1, a_2\}$ and two agents $N = \{1, 2\}$ where valuation functions are described in Table~\ref{table:PROP-Max-nega}.
The allocation in which each agent receives one item is EF, but not PROP-Max for agent~$1$.
\begin{table}[h]
\centering
\begin{tabular}{l|c|c}
 \hline
 & $a_1$ & $a_2$
 \\
 \hline
 $1$ & $-1 $ & $-1$ \\
  \hline
$2$ & $0 $ & $0$  \\
 \hline
\end{tabular}
\caption{Values for Items}
\label{table:PROP-Max-nega}
\end{table}
\end{proof}

\citet{SSG21a} proposed the following definition of Proportionality-Average (PROP-Ave).\footnote{
These authors called the property \emph{average-share}.}

\begin{definition}[PROP-Ave]
\label{def:PROP_ave}
Given an allocation $\pi$, the Proportional-Average share of agent $i$ (PROP-Ave$_i$) is
\[
\text{PROP-Ave}_i = \frac{1}{n} \sum_{a\in A} \sum_{j \in N} V_i(j, a).
\]
An allocation $\pi$ satisfies Proportionality-Ave (PROP-Ave) if $V_i(\pi) \geq$ PROP-Ave$_i$ holds for each agent $i \in N$.
\end{definition}

For $i,j\in N$ and $A'\subseteq A$, let $V_i(j,A') = \sum_{a\in A'}V_i(j,a)$.

\begin{proposition}
\label{prop:EF&PROP-Ave}
EF implies PROP-Ave for two agents.
\end{proposition}

\begin{proof}
Given a set of items $A$, consider any EF allocation $\pi = \pi_1 \cup \pi_2$ in which agent $1$ receives bundle $\pi_1$ and agent~$2$ receives bundle $\pi_2$. 
For agent $1$, by the definition of EF, we have $V_1(1, \pi_1) + V_1(2, \pi_2)\geq V_1(2, \pi_1) + V_1(1, \pi_2)$. 
Hence
\begin{align*}
  &  (V_1(1, \pi_1) + V_1(2, \pi_2)) + (V_1(1, \pi_1) + V_1(2, \pi_2))   
\\
&\geq   (V_1(1, \pi_1) + V_1(2, \pi_2)) + (V_1(2, \pi_1) + V_1(1, \pi_2)) 
\\
&=   \sum_{a\in A} (V_1(1, a) + V_1(2, a)).
\end{align*}
Thus allocation $\pi$ also satisfies PROP-Ave for agent~$1$. We can use the same argument for agent~$2$.
This completes the proof of Proposition~\ref{prop:EF&PROP-Ave}.
\end{proof}

\begin{proposition}
\label{prop:EF&PROP}
EF implies neither PROP-Max nor PROP-Ave for three agents even when all values and externalities are nonnegative.
\end{proposition}

\begin{proof}
Consider three agents $N = \{1, 2, 3\}$ and one item $A = \{a\}$, where 
$V_i(j, a) = 0$ for any $i,j \in N$ except that $V_1(2, a) = 1$.
The allocation in which agent $3$ receives item $a$ is EF, but it is neither PROP-Max nor PROP-Ave for agent~$1$.
\end{proof}

\begin{proposition}
\label{Prop:Prop-Ave&GFS}
Prop-Ave implies GFS.
\end{proposition}

\begin{proof}
Consider any PROP-Ave allocation $\pi$. For each agent~$i$, by definition of PROP-Ave, we have 
\begin{align*}
V_i(\pi) & \geq \frac{1}{n} \sum_{a\in A} \sum_{j\in N} V_i(j, a)
\\
& \geq \sum_{a\in A}\frac{V_i^{max}(a) + (n-1) V_i^{min}(a)}{n} 
\\
& = \text{GFS}_i 
\end{align*}
Thus any Prop-Ave allocation $\pi$ also satisfies GFS. 
\end{proof}

\begin{proposition}
\label{Prop:GFS&Max}
GFS implies Prop-Max if for each agent $i \in N$, $\sum_{a\in A} V_i^{min}(a) \geq 0$ holds.
\end{proposition}

\begin{proof}
Consider any GFS allocation $\pi$. For each agent $i$, by definition of GFS we have
\begin{align*}
V_i(\pi) & \geq \text{GFS}_i
\\
& = \frac{1}{n} \sum_{a\in A} V_i^{max}(a) +  \frac{n-1}{n} \sum_{a\in A} V_i^{min}(a)
\\
& \geq \frac{1}{n} \sum_{a\in A} V_i^{max}(a).
\end{align*}
Thus any GFS allocation $\pi$ also satisfies Prop-Max if for each agent $i$, we have $\sum_{a\in A} V_i^{min}(a) \geq 0$. 
\end{proof}

\subsection{Extended Maximin Share}

Maximin Share is a relaxation of proportionality for fair division of indivisible items, introduced by \citet{Budi11a}. The general idea is that an agent $i$ is asked to divide all items into $n$ disjoint bundles and must take the bundle with her minimum value. As a result, agent $i$ needs to partition all items in a way that maximizes the value of her minimum bundle. \citet{KPW18a} showed that in general there may not exist an allocation that guarantees maximin share for all agents. 

\citet{SSG21a} generalized maximin share to the case of externalities as shown in Definition~\ref{def:EMMS}. Let $P = \left \langle P_1, P_2, \dots, P_n \right \rangle$ denote a partition of items $A$ into $n$ bundles and let $\Omega_P$ denote the set of $n!$ different permutations of $P$. For agent $i$, let $W_i(P)$ denotes the worst allocation of $P$ for agent $i$, i.e., 
\[W_i(P) = \arg \min_{\pi \in \Omega_P} V_i(\pi).\]

\begin{definition}[EMMS]
\label{def:EMMS}
The extended maximin share of agent $i$ (EMMS$_i$) is denoted by 
\[\text{EMMS}_i = \max_{P \in \mathcal{P}} V_i(W_i(P))\]
where $\mathcal{P}$ is the set of all partitions of $A$ into $n$ subsets. 
An allocation $\pi$ satisfies extended maximin share if for all $i \in N$, $V_i(\pi) \geq $ EMMS$_i$ holds.
\end{definition}

\citet{SSG21a} proved that PROP-Ave$_i$ $\geq$ EMMS$_i$ for every agent $i$.
Even though their model assumes nonnegative externalities, their proof works even when externalities can be negative.
We therefore have the following proposition.

\begin{proposition}[\citet{SSG21a}]
\label{prop:PROP-Ave&EMMS}
PROP-Ave implies EMMS. 
\end{proposition}

\subsection{Partial Proportionality}
For fair division of indivisible items without externalities, EF implies proportionality. However, as we have seen, EF implies neither PROP-Max nor PROP-Ave when externalities exist. We propose a new definition that connects both EF and PROP-Ave. The intuition is that for any subset of agents $N' \subseteq N$ with $|N'| \leq k$ (for some given $k$), each agent $i\in N'$ should receive at least $1/|N'|$ of the total value she can receive from all the items assigned to the group $N'$.

\begin{definition}[$k$-P-PROP]
\label{def:k-P-PROP}
Given an allocation $\pi$, for a set of agents $N' \subseteq N$ with $|N'| = k$, 
the Partial-Proportional share of agent $i \in N$ with respect to the group $N'$ (P-PROP$_i^{N'}$) is
\[
\text{P-PROP}_i^{N'} = \frac{1}{|N'|} V_i^{sum}(\pi_{N'})
\]
where $V_i^{sum}(\pi_{N'}) = \sum_{a \in \pi_{N'}}
\sum_{j \in N'} V_i(j, a)$ denotes the sum of values agent $i$ derives from the items belonging to agents from the group $N'$ in $\pi$ when these items are assigned to other agents from $N'$.

An allocation $\pi$ is k-Partial-Proportional (k-P-PROP) if for any agent $i$, for any set of agents $N' \subseteq N$ with $|N'| \leq k$ and $i\in N'$, we have $V_i(\pi) \geq$ PROP$_i^{N'}$.
\end{definition}

It follows directly from the definitions that $n$-P-PROP implies PROP-Ave and 2-P-PROP. Next, we show a relation between 2-P-PROP and EF.

\begin{proposition}
\label{prop:EF&2-P-PROP}
EF implies 2-P-PROP. 
\end{proposition}

\begin{proof}
Consider any EF allocation $\pi$ in which each agent~$i$ receives bundle $\pi_i$. By the definition of EF, for any two agents $i$ and $j$, we have 
\[
V_i(i, \pi_i) + V_i(j, \pi_j)\geq V_i(j, \pi_i) + V_i(i, \pi_j).
\] 
We add $V_i(i, \pi_i) + V_i(j, \pi_j)$ to both sides:
\begin{align*}
  (V_i(i, &\pi_i) + V_i(j, \pi_j)) + (V_i(i, \pi_i) + V_i(j, \pi_j))
\\
&\geq  (V_i(i, \pi_i) + V_i(j, \pi_j)) + (V_i(j, \pi_i) + V_i(i, \pi_j)) 
\\
&=  V_i(i, \pi_i \cup \pi_j) + V_i(j, \pi_i \cup \pi_j)
\\
&=  \sum_{a\in \pi_i \cup \pi_j} (V_i(i, a) + V_i(j, a)).
\end{align*}
It follows that
\[
V_i(i, \pi_i) + V_i(j, \pi_j) \geq \frac{1}{2} \sum_{a\in \pi_i \cup \pi_j} (V_i(i, a) + V_i(j, a))
\]
which means that $\pi$ satisfies 2-P-PROP.
\end{proof}

\section{Conclusion}
In this paper, we proposed several fairness concepts for fair division of indivisible items under externalities including EF1, EFX and GFS. We presented efficient algorithms for finding the corresponding fair allocations. 
An important open question that remains from our work is whether there always exists an EF1 allocation among three or more agents.
Note that a positive answer to this question would generalize the corresponding result of \citet{ACIW19a} for goods and chores without externalities.
On the other hand, if the answer is negative, it would be reasonable to ask for the optimal relaxation EF$k$ that can be attained.
Finally, it will be interesting to consider more general valuation functions that are not necessarily additive. 
While the existence of EF1 and EFX allocations is known beyond additive valuations in certain settings \citep{LMMS04a,PlRo20a}, it remains to be seen whether these guarantees can be extended to incorporate externalities.

\section*{Acknowledgments}

This work was partially supported by ARC Laureate Project FL200100204 on ``Trustworthy AI'', by the Singapore Ministry of Education under grant
number MOE-T2EP20221-0001, and by an NUS Start-up Grant.
We thank the anonymous reviewers of AAAI 2023 for their comments.

\bibliography{aaai23}

\begin{thebibliography}{31}
\providecommand{\natexlab}[1]{#1}

\bibitem[{Agarwal et~al.(2021)Agarwal, Elkind, Gan, Igarashi, Suksompong, and
  Voudouris}]{AEGI21a}
Agarwal, A.; Elkind, E.; Gan, J.; Igarashi, A.; Suksompong, W.; and Voudouris,
  A.~A. 2021.
\newblock Schelling games on graphs.
\newblock \emph{Artificial Intelligence}, 301: 103576.

\bibitem[{Aleksandrov et~al.(2015)Aleksandrov, Aziz, Gaspers, and
  Walsh}]{AAGW15b}
Aleksandrov, M.; Aziz, H.; Gaspers, S.; and Walsh, T. 2015.
\newblock Online fair division: Analysing a food bank problem.
\newblock In \emph{Proceedings of the 24th International Joint Conference on
  Artificial Intelligence (IJCAI)}, 2540--2546.

\bibitem[{Amanatidis et~al.(2022)Amanatidis, Birmpas, Filos-Ratsikas, and
  Voudouris}]{ABF+22a}
Amanatidis, G.; Birmpas, G.; Filos-Ratsikas, A.; and Voudouris, A.~A. 2022.
\newblock Fair division of indivisible goods: a survey.
\newblock In \emph{Proceedings of the 31st International Joint Conference on
  Artificial Intelligence (IJCAI)}, 5385--5393.

\bibitem[{Aziz et~al.(2018)Aziz, Bouveret, Caragiannis, Giagkousi, and
  Lang}]{ABC+18a}
Aziz, H.; Bouveret, S.; Caragiannis, I.; Giagkousi, I.; and Lang, J. 2018.
\newblock Knowledge, fairness, and social constraints.
\newblock In \emph{Proceedings of the 32nd AAAI Conference on Artificial
  Intelligence (AAAI)}, 4638--4645.

\bibitem[{Aziz et~al.(2022{\natexlab{a}})Aziz, Caragiannis, Igarashi, and
  Walsh}]{ACIW19a}
Aziz, H.; Caragiannis, I.; Igarashi, A.; and Walsh, T. 2022{\natexlab{a}}.
\newblock Fair allocation of indivisible goods and chores.
\newblock \emph{Autonomous Agents and Multi-Agent Systems}, 36(1): 3:1--3:21.

\bibitem[{Aziz et~al.(2022{\natexlab{b}})Aziz, Li, Moulin, and Wu}]{ALM+22a}
Aziz, H.; Li, B.; Moulin, H.; and Wu, X. 2022{\natexlab{b}}.
\newblock Algorithmic fair allocation of indivisible items: a survey and new
  questions.
\newblock \emph{ACM SIGecom Exchanges}, 20(1): 24--40.

\bibitem[{Barman, Krishnamurthy, and Vaish(2018)}]{BKV18a}
Barman, S.; Krishnamurthy, S.~K.; and Vaish, R. 2018.
\newblock Greedy algorithms for maximizing {N}ash social welfare.
\newblock In \emph{Proceedings of the 17th International Conference on
  Autonomous Agents and MultiAgent Systems (AAMAS)}, 7--13.

\bibitem[{Bouveret, Chevaleyre, and Maudet(2016)}]{BCM15a}
Bouveret, S.; Chevaleyre, Y.; and Maudet, N. 2016.
\newblock Fair allocation of indivisible goods.
\newblock In Brandt, F.; Conitzer, V.; Endriss, U.; Lang, J.; and Procaccia,
  A.~D., eds., \emph{Handbook of Computational Social Choice}, chapter~12,
  284--310. Cambridge University Press.

\bibitem[{Bouveret and Lema{\^\i}tre(2016)}]{BoLe15a}
Bouveret, S.; and Lema{\^\i}tre, M. 2016.
\newblock Characterizing conflicts in fair division of indivisible goods using
  a scale of criteria.
\newblock \emph{Autonomous Agents and Multi-Agent Systems}, 30(2): 259--290.

\bibitem[{Brams and Taylor(1996)}]{BrTa96a}
Brams, S.~J.; and Taylor, A.~D. 1996.
\newblock \emph{Fair Division: From Cake-Cutting to Dispute Resolution}.
\newblock Cambridge University Press.

\bibitem[{Br{\^a}nzei, Procaccia, and Zhang(2013)}]{SPZ13a}
Br{\^a}nzei, S.; Procaccia, A.~D.; and Zhang, J. 2013.
\newblock Externalities in cake cutting.
\newblock In \emph{Proceedings of the 23rd International Joint Conference on
  Artificial Intelligence (IJCAI)}, 55--61.

\bibitem[{Budish(2011)}]{Budi11a}
Budish, E. 2011.
\newblock The combinatorial assignment problem: Approximate competitive
  equilibrium from equal incomes.
\newblock \emph{Journal of Political Economy}, 119(6): 1061--1103.

\bibitem[{Caragiannis et~al.(2019)Caragiannis, Kurokawa, Moulin, Procaccia,
  Shah, and Wang}]{CKM+19a}
Caragiannis, I.; Kurokawa, D.; Moulin, H.; Procaccia, A.~D.; Shah, N.; and
  Wang, J. 2019.
\newblock The unreasonable fairness of maximum {N}ash welfare.
\newblock \emph{ACM Transactions on Economics and Computation}, 7(3):
  12:1--12:32.

\bibitem[{Chaudhury, Garg, and Mehlhorn(2020)}]{CGM20a}
Chaudhury, B.~R.; Garg, J.; and Mehlhorn, K. 2020.
\newblock {EFX} exists for three agents.
\newblock In \emph{Proceedings of the 21st ACM Conference on Economics and
  Computation (ACM-EC)}, 1--19.

\bibitem[{Chauhan, Lenzner, and Molitor(2018)}]{CLM18a}
Chauhan, A.; Lenzner, P.; and Molitor, L. 2018.
\newblock Schelling segregation with strategic agents.
\newblock In \emph{Proceedings of the 11th International Symposium on
  Algorithmic Game Theory (SAGT)}, 137--149.

\bibitem[{Conitzer, Freeman, and Shah(2017)}]{CFS17a}
Conitzer, V.; Freeman, R.; and Shah, N. 2017.
\newblock Fair public decision making.
\newblock In \emph{Proceedings of the 18th ACM Conference on Economics and
  Computation (ACM-EC)}, 629--646.

\bibitem[{Darmann and Schauer(2015)}]{DaSc15a}
Darmann, A.; and Schauer, J. 2015.
\newblock Maximizing {N}ash product social welfare in allocating indivisible
  goods.
\newblock \emph{European Journal of Operational Research}, 247(2): 548--559.

\bibitem[{Elkind et~al.(2020)Elkind, Patel, Tsang, and Zick}]{EPTZ20a}
Elkind, E.; Patel, N.; Tsang, A.; and Zick, Y. 2020.
\newblock Keeping your friends close: Land allocation with friends.
\newblock In \emph{Proceedings of the 29th International Joint Conference on
  Artificial Intelligence (IJCAI)}, 318--324.

\bibitem[{Freeman et~al.(2019)Freeman, Sikdar, Vaish, and Xia}]{FSVX19a}
Freeman, R.; Sikdar, S.; Vaish, R.; and Xia, L. 2019.
\newblock Equitable allocations of indivisible goods.
\newblock In \emph{Proceedings of the 28th International Joint Conference on
  Artificial Intelligence (IJCAI)}, 280--286.

\bibitem[{Gross-Humbert et~al.(2021)Gross-Humbert, Benabbou, Beynier, and
  Maudet}]{GBBM21a}
Gross-Humbert, N.; Benabbou, N.; Beynier, A.; and Maudet, N. 2021.
\newblock Sequential and swap mechanisms for public housing allocation with
  quotas and neighbourhood-based utilities.
\newblock In \emph{Proceedings of the 20th International Conference on
  Autonomous Agents and Multiagent Systems (AAMAS)}, 1521--1523.

\bibitem[{Halpern et~al.(2020)Halpern, Procaccia, Psomas, and Shah}]{HPPS20a}
Halpern, D.; Procaccia, A.~D.; Psomas, A.; and Shah, N. 2020.
\newblock Fair division with binary valuations: One rule to rule them all.
\newblock In \emph{Proceedings of the 16th International Conference on Web and
  Internet Economics (WINE)}, 370--383.

\bibitem[{Kurokawa, Procaccia, and Wang(2018)}]{KPW18a}
Kurokawa, D.; Procaccia, A.~D.; and Wang, J. 2018.
\newblock Fair enough: Guaranteeing approximate maximin shares.
\newblock \emph{Journal of the ACM}, 64(2): 8:1--8:27.

\bibitem[{Li, Zhang, and Zhang(2015)}]{LZZ15a}
Li, M.; Zhang, J.; and Zhang, Q. 2015.
\newblock Truthful cake cutting mechanisms with externalities: Do not make them
  care for others too much!
\newblock In \emph{Proceedings of the 24th International Joint Conference on
  Artificial Intelligence (IJCAI)}, 589--595.

\bibitem[{Lipton et~al.(2004)Lipton, Markakis, Mossel, and Saberi}]{LMMS04a}
Lipton, R.~J.; Markakis, E.; Mossel, E.; and Saberi, A. 2004.
\newblock On approximately fair allocations of indivisible goods.
\newblock In \emph{Proceedings of the 5th ACM Conference on Economics and
  Computation (ACM-EC)}, 125--131.

\bibitem[{Massand and Simon(2019)}]{MaSi19a}
Massand, S.; and Simon, S. 2019.
\newblock Graphical one-sided markets.
\newblock In \emph{Proceedings of the 28th International Joint Conference on
  Artificial Intelligence (IJCAI)}, 492--498.

\bibitem[{Mishra, Padala, and Gujar(2022)}]{MPG21a}
Mishra, S.; Padala, M.; and Gujar, S. 2022.
\newblock Fair allocation with special externalities.
\newblock In \emph{Proceedings of the 19th Pacific Rim International Conference
  on Artificial Intelligence (PRICAI)}, 3--16.

\bibitem[{Plaut and Roughgarden(2020)}]{PlRo20a}
Plaut, B.; and Roughgarden, T. 2020.
\newblock Almost envy-freeness with general valuations.
\newblock \emph{SIAM Journal on Discrete Mathematics}, 34(2): 1039--1068.

\bibitem[{Seddighin, Saleh, and Ghodsi(2021)}]{SSG21a}
Seddighin, M.; Saleh, H.; and Ghodsi, M. 2021.
\newblock Maximin share guarantee for goods with positive externalities.
\newblock \emph{Social Choice and Welfare}, 56(2): 291--324.

\bibitem[{Suksompong and Teh(2022)}]{SuTe22a}
Suksompong, W.; and Teh, N. 2022.
\newblock On maximum weighted {N}ash welfare for binary valuations.
\newblock \emph{Mathematical Social Sciences}, 117: 101--108.

\bibitem[{Thomson(2016)}]{Thom15a}
Thomson, W. 2016.
\newblock Introduction to the theory of fair allocation.
\newblock In Brandt, F.; Conitzer, V.; Endriss, U.; Lang, J.; and Procaccia,
  A.~D., eds., \emph{Handbook of Computational Social Choice}, chapter~11,
  261--283. Cambridge University Press.

\bibitem[{Velez(2016)}]{Vele16a}
Velez, R.~A. 2016.
\newblock Fairness and externalities.
\newblock \emph{Theoretical Economics}, 11(1): 381--410.

\end{thebibliography}

\appendix

\section{Proof of \Cref{cor:EF1-two}}
\label{app:proof-cor-EF1-two}

In this section, we prove that there always exists an EF1 allocation between two agents which can be computed in linear time. Similar to the proof of \Cref{theo:EFX:two}, we first prove a simpler case between two agents with symmetric valuations in Lemma~\ref{lemma:2same1} and then complete the proof for \Cref{cor:EF1-two}.

\begin{lemma}
\label{lemma:2same1}
An EF1 allocation always exists for two agents with symmetric valuations and it can be computed in time $O(m)$.
\end{lemma}

\begin{proof}
Consider a set of items $A$ and two agents $N = \{1, 2\}$ who have symmetric additive valuations. 
Create an allocation $\tilde{\pi}$ between agent $1$ and $2$ as follows. We iteratively and greedily allocate each item. There are two possible bundles for the item, leading to two different allocations. Between these two allocations, we choose one that an agent with the smaller current total value weakly prefers. Break ties arbitrarily.

We next prove that allocation $\tilde{\pi}$ is EF1 for both agents. 
Suppose we allocate all items in the order of $a_1, a_2, \ldots, a_m$.
For the base case, assigning item $a_1$ to either agent is EF1. 
For the induction, assume that a partial allocation of items $a_1$, $\ldots$ , $a_{k}$ is EF1. Since the agents have symmetric valuations, at most one agent can be envious. Without loss of generality, suppose the algorithm allocates $a_{k+1}$ to agent $1$ who has at most the same value as agent~$2$. Then agent $1$'s envy towards agent $2$ weakly decreases and the allocation is still EF1 for agent $1$. If agent $2$ becomes envious, then removing the new item $a_{k+1}$ will eliminate the envy. Thus the allocation remains EF1 for both agent $1$ and $2$. 

We can assign all items in linear time and thus we can compute an EF1 allocation between two agents with symmetric valuations in linear time. This completes the proof of Lemma~\ref{lemma:2same1}.
\end{proof}

\begin{corollary}
\label{coro:EF1:two}
There always exists an EF1 allocation between two agents which can be computed in time $O(m)$.
\end{corollary}

\begin{proof}
First create a dummy agent $1'$ of $1$. Both agent $1$ and $1'$ treat each other as agent $2$ and they have a symmetric valuation function such that for any item $a \in A$, we have $V_1(1, a) = V_{1'}(1', a)$ and $V_1(1', a) = V_{1'}(1, a) = V_1(2, a)$.
That is, if item $a$ is assigned to $1'$, then agent $1'$ receives the value $V_1(1, a)$ and agent $1$ receives the value $V_1(2, a)$ as if item $a$ is assigned to $2$ from the perspective of agent $1$. 

Compute an EF1 allocation $\tilde{\pi}$ between agents $1$ and $1'$ via the algorithm in the proof of Lemma~\ref{lemma:2same1}. 
Allocation $\tilde{\pi}$ divides all items $A$ into two bundles; let agent $2$ first choose the bundle she prefers and leave the remaining bundle to agent $1$. Since agent $2$ chooses first, she does not envy agent~$1$. We showed that $\tilde{\pi}$ is EF1 between $1$ and $1'$ in Lemma~\ref{lemma:2same1}, so it is EF1 no matter which bundle agent~$1$ receives. This completes the proof of Corollary~\ref{coro:EF1:two}.
\end{proof}

\section{EF1 for Three Agents under Binary Valuations and No-Chore Assumption}
\label{app:three-no-chore}

In this section, we give a detailed description of our algorithm that computes an EF1 allocation among three agents under the no-chore assumption and binary valuations. 

\subsection{Difference of Valuations between Two Agents}
We first introduce a new notation to simplify the description of the algorithm. 
Given an item $a$, let $\Delta_{ij}(a) = V_i(i, a) - V_i(j, a)$ denote the difference between the value agent $i$ receives when item $a$ is assigned to agent $i$ and the value agent $i$ receives when item $a$ is assigned to agent $j$.
Since valuations are binary and satisfy the no-chore assumption, we have $\Delta_{ij}(a)\in\{0,1\}$ for all $i,j\in N$ and $a\in A$.
We extend the notation to sets of items as follows:
\[
\Delta_{ij}(A') = \sum_{a \in A'} \Delta_{ij}(a).
\]
An allocation $\pi$ is EF if for each pair $i, j\in N$, we have 
\[\Delta_{ij}(\pi_i) - \Delta_{ij}(\pi_j)\geq 0.\]
An allocation $\pi$ is EF1 if for each pair $i, j\in N$, if $\Delta_{ij}(\pi_i) - \Delta_{ij}(\pi_j) < 0$, then there exists $a\in A$ such that
\[\Delta_{ij}(\pi_i \setminus \{a\}) - \Delta_{ij}(\pi_j \setminus \{a\})\geq 0.\]
We write $\Delta_{ij}(\pi) = \Delta_{ij}(\pi_i) - \Delta_{ij}(\pi_j)$.

\begin{definition}[Items of the Same Type]
We say that two items $a$, $b$ $\in$ $A$ belong to the same type if for any pair $i, j \in N$, we have $\Delta_{ij}(a) = \Delta_{ij}(b)$. 
\end{definition}

\subsection{Representation of Valuation Functions}

Given an item $a \in A$, we can create a matrix to represent the agents' valuation functions as follows.
\[
\begin{matrix}
c^{11} & c^{12} & \cdots & c^{1i} & \cdots & c^{1n} \\
c^{21} & c^{22} & \cdots & c^{2i} & \cdots & c^{2n} \\
\vdots & \vdots & \ddots & \vdots & \ddots & \vdots \\
c^{n1} & c^{n2} & \cdots & c^{ni} & \cdots & c^{nn} 
\end{matrix}
\]
\begin{itemize}
  \item Each entry $c^{ij}$ represents the value agent $i$ receives when item $a$ is assigned to agent $j$, i.e., $c^{ij} = V_i(j, a)$.
\item The $i$th row corresponds to the values agent $i$ receives when item $a$ is assigned to each agent $j \in N$.
\item The $j$th column corresponds to the values each agent $i \in N$ receives when item $a$ is assigned to agent $j$.
\end{itemize}

Note that two items belong to the same type if they correspond to the same matrix.

\subsection{Characterization of EF1 Allocations}

In this section, we characterize some features of EF1 allocations based on $\Delta_{ij}$. 
We will design an efficient algorithm based on the following lemmas.

\begin{lemma}
\label{lemma:no-more-than-2}
Given an instance $I = (N, A, V)$, suppose there exists a subset of items $A^t$ of the same type $t$ whose size is divisible by $n$.
Instance $I$ admits an EF1 allocation if the reduced instance $I' = (N, A', V)$ with $A' = A \setminus A^t$ admits an EF1 allocation. 
\end{lemma}

\begin{proof}
Suppose that $|A^t| = k\cdot n$.
Note that assigning each agent $k$ items of type $t$ does not incur envy between any two agents. If the reduced instance admits an EF1 allocation $\pi'$, then the original instance $I$ also admits an EF1 allocation $\pi$ which is obtained from $\pi'$ by additionally assigning each agent $k$ items of type $t$.
\end{proof}

\begin{lemma}
\label{Lemma:no-envy-at-all}
Given an instance $I = (N, A, V)$, suppose there exists an item $a \in A$ such that for some agent $i$, $\Delta_{ji}(a) = 0$ holds for all $j \in N$. 
Instance $I$ admits an EF1 allocation if the reduced instance $I' = (N, A', V)$ with $A' = A \setminus \{a\}$ admits an EF1 allocation.
\end{lemma}

\begin{proof}
Assigning item $a$ to agent $i$ does not incur envy towards $i$ from any agent $j \in N$. If the reduced instance $I'$ admits an EF1 allocation $\pi'$, then the original instance $I$ also admits an EF1 allocation $\pi$ which is obtained from $\pi'$ by additionally assigning  item $a$ to agent $i$. 
\end{proof}

\begin{lemma}
\label{Lemma:no-envy-from-one}
Given an instance $I = (N, A, V)$ among three agents under the no-chore assumption and binary valuations,
and assume there exists an agent $i$ $\in$ $N$ such that for each item $a\in A$, $\Delta_{ij}(a) = 0$ holds for all $j$ $\in$ $N$. Then instance $I$ admits an EF1 allocation. 
\end{lemma}
\begin{proof}
For the instance $I$, we have one agent $i$ such that assigning any item $a\in A$ to the other two agents does not incur envy from agent $i$. Thus we can consider an equivalent problem of allocating all items between the other two agents only. Since there always exists an EF1 allocation between two agents, this completes the proof.
\end{proof}

\begin{lemma}
\label{lemma:merge}
Given an instance $I = (N, A, V)$ among three agents under the no-chore assumption and binary valuations, 
let $\pi$ denote an EF1 allocation in which for each pair $i$, $j$ $\in$ $N$, $i$ and $j$ do not envy each other simultaneously, that is, if $\Delta_{ij}(\pi)$ $=$ $-1$ holds, then we have $\Delta_{ji}(\pi)$ $\geq$ $0$.

Consider a new item $a \notin A$ for which there exists an agent $i$ $\in$ $N$ such that $\Delta_{ij}(a) = 0$ holds for all $j$ $\in$ $N$. Then the new instance $I' = (N, A \cup \{a\}, V)$ also admits an EF1 allocation.
\end{lemma}

\begin{proof}
For item $a$, there exists an agent, say agent $1$, such that $\Delta_{1j}(a) = 0$ holds for all $j$ $\in$ $N$. Then we can assign item $a$ to either agent $2$ or $3$ without incurring envy from agent $1$. 
If agent $2$ and $3$ do not envy each other in the allocation $\pi$, then assigning $a$ to either $2$ or $3$ results in an EF1 allocation. 
Suppose agent $2$ envies $3$ in the allocation $\pi$. Since $\pi$ is an EF1 allocation in which any two agents do not envy each other simultaneously, we can assign $a$ to agent $2$, and agent $3$ will not envy agent $2$ by more than one item. Thus the new allocation is still EF1.
\end{proof}

We will show in Proposition~\ref{prop:envy_simultaneously} that there always exists an EF1 allocation in which any two agents do not envy each other simultaneously after applying our reduction rules.

\subsection{Reduction Rules that Remove One Item Each Time}
\label{sec:reduction}

Next we design three reduction rules that remove one item each time in an EF or EF1 way. 

The first rule captures the idea that if assigning some item $a$ to agent $i$ does not incur envy from any agent, then it is safe to assign item $a$ to agent $i$.

\begin{reduction}
\label{reduction:column-1}
Remove any item $a$ with a corresponding matrix $M$ such that there exists a column of $1$, i.e., $M_{1j} = M_{2j} = M_{3j} = 1$.
\end{reduction}
\begin{proof}
Since we have $\Delta_{ij}(a) = 0$ for all $i \in \{1, 2, 3\}$, assigning item $a$ to agent $j$ does not incur envy from any other agent by Lemma~\ref{Lemma:no-envy-at-all}.
\end{proof}
After applying Reduction Rule~\ref{reduction:column-1}, for each type of items, each column in 
its corresponding matrix contains at least one 0.

The next two rules capture the idea that for items $a$ such that assigning $a$ to any agent would not incur envy from agent~$i$, but assigning $a$ to either of the two remaining agents would incur envy from the other agent, then we can divide such items equally between the latter two agents, possibly leaving one item if the number of such items is odd.
We can allocate the remainder item between the latter two agents in an EF1 way at the end. 

\begin{reduction}
\label{reduction:diagonal-0}
Remove any item $a$ with a corresponding matrix ${M}$ such that $M_{11}$, $M_{22}$ or $M_{33}$ equals 0.
\end{reduction}

\begin{proof}
By the no-chore assumption, if $M_{11}$, $M_{22}$ or $M_{33}$ equals 0, then we have $\Delta_{ij}(a) = 0$ for some $i \in \{1, 2, 3\}$ and all $j \in \{1, 2, 3\}$. Thus assigning item $a$ between the other two agents will not incur envy from agent $i$. 
We can group $a$ with all items that do not incur envy from agent $i$ no matter who it is allocated to, but incurs envy from the remaining two agents when it is allocated to the other agent in the pair.
If there are at least two such items, allocating one each to the latter two agents does not incur any envy.
Hence, we will be left with at most one such item.
By Lemma~\ref{lemma:merge}, we can allocate this item at the end after applying all other reduction rules.
\end{proof}

After applying Reduction Rule~\ref{reduction:diagonal-0}, for each type of items, the diagonal of 
its corresponding matrix consists of $1$ only.

\begin{reduction}
\label{reduction:row-1}
Remove any item $a$ with a corresponding matrix $M$ such that there exists a row $i$ with $M_{i1} = M_{i2} = M_{i3} = 1$.
\end{reduction}

\begin{proof}
For item $a$, we have $\Delta_{ij}(a) = 0$ for some agent $i \in \{1, 2, 3\}$ and all $j \in \{1, 2, 3\}$. Assigning $a$ to the other two agents will not incur envy from agent $i$. We can remove this item from the instance now, and allocate item $a$ between the other two agents in an EF1 way later, as discussed in the proof of Reduction Rule~\ref{reduction:diagonal-0}.
\end{proof}

After applying Reduction Rule~\ref{reduction:row-1}, for each type of items, each row in 
its corresponding matrix contains at least one 0.

\subsection{18 Cases Left after Applying Three Reduction Rules that Remove One Item Each Time}

After applying the three reduction rules that remove one item each time, there are 18 different types of items left. We can classify all remaining types into different categories based on the number of rows with two zeros. Let $x_{B}$ denote the set of matrices having two zeros in the rows $B\subseteq\{1,2,3\}$. 
For simplicity, we write, e.g., $x_{123}$ instead of $x_{\{1,2,3\}}$.
The superscripts are used to distinguish between matrices of the same category.

\begin{itemize}
\item $x_{123}$: all three rows have two zeros each
\end{itemize}

$$
\begin{bmatrix}
1 & 0 & 0 \\
0 & 1 & 0 \\
0 & 0 & 1 
\end{bmatrix}
x_{123}^0
$$
\begin{itemize}
\item $x_{12}, x_{13}, x_{23}$: two rows have two zeros each
\end{itemize}
\begin{align*}
&
\begin{bmatrix}
1 & 0 & 0 \\
0 & 1 & 0 \\
0 & 1 & 1
\end{bmatrix}
x_{12}^0
\quad
\begin{bmatrix}
1 & 0 & 0 \\
0 & 1 & 0 \\
1 & 0 & 1
\end{bmatrix}
x_{12}^1
\\
&
\begin{bmatrix}
1 & 0 & 0 \\
0 & 1 & 1 \\
0 & 0 & 1
\end{bmatrix}
x_{13}^0
\quad
\begin{bmatrix}
1 & 0 & 0 \\
1 & 1 & 0 \\
0 & 0 & 1
\end{bmatrix}
x_{13}^1
\\
&
\begin{bmatrix}
1 & 0 & 1 \\
0 & 1 & 0 \\
0 & 0 & 1
\end{bmatrix}
x_{23}^0
\quad
\begin{bmatrix}
1 & 1 & 0 \\
0 & 1 & 0 \\
0 & 0 & 1
\end{bmatrix}
x_{23}^1
\end{align*}

\begin{itemize}
\item $x_{1}, x_{2}, x_{3}$: one row has two zeros
\end{itemize}
\begin{align*}
&
\begin{bmatrix}
1 & 0 & 0 \\
0 & 1 & 1 \\
0 & 1 & 1
\end{bmatrix}
x_{1}^0
\quad
\begin{bmatrix}
1 & 0 & 0 \\
0 & 1 & 1 \\
1 & 0 & 1
\end{bmatrix}
x_{1}^1
\quad
\begin{bmatrix}
1 & 0 & 0 \\
1 & 1 & 0 \\
0 & 1 & 1
\end{bmatrix}
x_{1}^2
\\
&
\begin{bmatrix}
1 & 0 & 1 \\
0 & 1 & 0 \\
0 & 1 & 1
\end{bmatrix}
x_{2}^0
\quad
\begin{bmatrix}
1 & 0 & 1 \\
0 & 1 & 0 \\
1 & 0 & 1
\end{bmatrix}
x_{2}^1
\quad
\begin{bmatrix}
1 & 1 & 0 \\
0 & 1 & 0 \\
1 & 0 & 1
\end{bmatrix}
x_{2}^2
\\
&
\begin{bmatrix}
1 & 0 & 1 \\
1 & 1 & 0 \\
0 & 0 & 1
\end{bmatrix}
x_{3}^0
\quad
\begin{bmatrix}
1 & 1 & 0 \\
0 & 1 & 1 \\
0 & 0 & 1
\end{bmatrix}
x_{3}^1
\quad
\begin{bmatrix}
1 & 1 & 0 \\
1 & 1 & 0 \\
0 & 0 & 1
\end{bmatrix}
x_{3}^2
\end{align*}
\begin{itemize}
\item $x_0$: no row has two zeros
\end{itemize}

$$
\begin{bmatrix}
1 & 0 & 1 \\
1 & 1 & 0 \\
0 & 1 & 1
\end{bmatrix}
x_{0}^0
\quad
\begin{bmatrix}
1 & 1 & 0 \\
0 & 1 & 1 \\
1 & 0 & 1
\end{bmatrix}
x_{0}^1
$$

\subsection{Reduction Rules That Remove A Pair of / Three Items Each Time}

Next, we introduce three more reduction rules that remove a pair of or three items each time in an EF manner. 

\begin{reduction}
\label{reduction:no-more-than-2}
For any type of items $A^t$, remove $3k$ items with $k = \left\lfloor \frac{|A^t|}{3}\right\rfloor$.
\end{reduction}

\begin{proof}
By Lemma~\ref{lemma:no-more-than-2}, we can assign each agent $k$ items of the same type without incurring envy. 
\end{proof}

After applying Reduction Rule~\ref{reduction:no-more-than-2}, the number of items of each type does not exceed $2$.

\begin{reduction}
\label{reduction:tuple}
If there are three items $a$, $b$ and $c$ with matrices as follows:
\begin{align*}
&
a:
\begin{bmatrix}
1 & 0 & 0 \\
X & 1 & X \\
X & X & 1
\end{bmatrix}
\quad
b:
\begin{bmatrix}
1 & X & X \\
0 & 1 & 0 \\
X & X & 1
\end{bmatrix}
\quad
c:
\begin{bmatrix}
1 & X & X \\
X & 1 & X \\
0 & 0 & 1
\end{bmatrix}
\end{align*}
where $X$ can be either $0$ or $1$. Then assigning $a$ to agent $1$, assigning $b$ to agent $2$ and assigning $c$ to agent $3$ will not incur envy.
\end{reduction}

\begin{proof}
It is easy to verify that this assignment will not incur envy from any agent.
\end{proof}

In other words, if we have three items of category $x_1$, $x_2$ and $x_3$ respectively, it is safe to remove them. Note that, for example, $x_{12}$ can be considered as either $x_1$ or $x_2$.

\begin{reduction}
\label{reduction:row-sum-1}
Remove two items $a$ and $b$ if assigning $a$ to agent $i$ and assigning $b$ to agent $j$
does not incur envy from any agent.
\end{reduction}
\begin{proof}
It follows directly from the statement.
\end{proof}

\begin{example}[Example of Reduction Rule~\ref{reduction:row-sum-1}]
Consider the following two items. If we assign item $a$ to agent $1$ and assign item $b$ to agent $2$, then it will not incur envy from any agent.
\begin{align*}
&
a:
\begin{bmatrix}
1 & 0 & 0 \\
0 & 1 & 0 \\
1 & 0 & 1
\end{bmatrix}
\quad
b:
\begin{bmatrix}
1 & 0 & 0 \\
0 & 1 & 0 \\
0 & 1 & 1
\end{bmatrix}
\end{align*}
\end{example}

\subsection{More on Reduction Rule~\ref{reduction:row-sum-1}}
After applying Reduction Rule~\ref{reduction:row-sum-1}, several matrices cannot coexist with each other. Here is a list of incompatible matrices which will be used later.

Within $x_{12}, x_{13}$:
\begin{align*}
&
\begin{bmatrix}
1 & 0 & 0 \\
0 & 1 & 0 \\
0 & 1 & 1
\end{bmatrix}
(x_{12}^0)
\Leftrightarrow
\begin{bmatrix}
1 & 0 & 0 \\
0 & 1 & 0 \\
1 & 0 & 1
\end{bmatrix}
(x_{12}^1)
\end{align*}
\begin{align*}
\begin{bmatrix}
1 & 0 & 0 \\
0 & 1 & 1 \\
0 & 0 & 1
\end{bmatrix}
(x_{13}^0)
\Leftrightarrow
\begin{bmatrix}
1 & 0 & 0 \\
1 & 1 & 0 \\
0 & 0 & 1
\end{bmatrix}
(x_{13}^1)
\end{align*}

Within $x_1$:
\begin{align*}
&
\begin{bmatrix}
1 & 0 & 0 \\
0 & 1 & 1 \\
1 & 0 & 1 
\end{bmatrix}
(x^{1}_{1})
\Leftrightarrow
\begin{bmatrix}
1 & 0 & 0 \\
0 & 1 & 1 \\
0 & 1 & 1
\end{bmatrix}
(x^{0}_{1})
\Leftrightarrow
\begin{bmatrix}
1 & 0 & 0 \\
1 & 1 & 0 \\
0 & 1 & 1
\end{bmatrix}
(x^{2}_{1})
\end{align*}

Within $x_2$:
\begin{align*}
\begin{bmatrix}
1 & 1 & 0 \\
0 & 1 & 0 \\
1 & 0 & 1
\end{bmatrix}
(x^{2}_{2})
\Leftrightarrow
\begin{bmatrix}
1 & 0 & 1 \\
0 & 1 & 0 \\
1 & 0 & 1
\end{bmatrix}
(x^{1}_{2})
\Leftrightarrow
\begin{bmatrix}
1 & 0 & 1 \\
0 & 1 & 0 \\
0 & 1 & 1
\end{bmatrix}
(x^{0}_{2})
\end{align*}

Between $x_1$ and $x_2$:
\begin{align*}
&
\begin{bmatrix}
1 & 0 & 0 \\
0 & 1 & 1 \\
1 & 0 & 1
\end{bmatrix}
(x^{1}_{1})
\Leftrightarrow
\begin{bmatrix}
1 & 0 & 1 \\
0 & 1 & 0 \\
0 & 1 & 1
\end{bmatrix}
(x^{0}_{2})
\end{align*}
\begin{align*}
\begin{bmatrix}
1 & 0 & 0 \\
0 & 1 & 1 \\
0 & 1 & 1
\end{bmatrix}
(x^{0}_{1})
\Leftrightarrow
\begin{bmatrix}
1 & 0 & 1 \\
0 & 1 & 0 \\
1 & 0 & 1
\end{bmatrix}
(x^{1}_{2})
\end{align*}

Between $x_{12}$ and $x_{1}$:
\begin{align*}
&
\begin{bmatrix}
1 & 0 & 0 \\
0 & 1 & 0 \\
1 & 0 & 1
\end{bmatrix}
(x^{1}_{12})
\Leftrightarrow
\begin{bmatrix}
1 & 0 & 0 \\
0 & 1 & 1 \\
0 & 1 & 1
\end{bmatrix}
(x^{0}_{1})
\end{align*}
\begin{align*}
\begin{bmatrix}
1 & 0 & 0 \\
0 & 1 & 0 \\
0 & 1 & 1
\end{bmatrix}
(x^{0}_{12})
\Leftrightarrow
\begin{bmatrix}
1 & 0 & 0 \\
0 & 1 & 1 \\
1 & 0 & 1
\end{bmatrix}
(x^{1}_{1})
\end{align*}

Between $x_{12}$ and $x_{2}$:
\begin{align*}
&
\begin{bmatrix}
1 & 0 & 0 \\
0 & 1 & 0 \\
1 & 0 & 1
\end{bmatrix}
(x^{1}_{12})
\Leftrightarrow
\begin{bmatrix}
1 & 0 & 1 \\
0 & 1 & 0 \\
0 & 1 & 1
\end{bmatrix}
(x^{0}_{2})
\end{align*}
\begin{align*}
\begin{bmatrix}
1 & 0 & 0 \\
0 & 1 & 0 \\
0 & 1 & 1
\end{bmatrix}
(x^{0}_{12})
\Leftrightarrow
\begin{bmatrix}
1 & 0 & 1 \\
0 & 1 & 0 \\
1 & 0 & 1
\end{bmatrix}
(x^{1}_{2})
\end{align*}

Between $x_{13}$ and $x_{1}$:
\begin{align*}
&
\begin{bmatrix}
1 & 0 & 0 \\
0 & 1 & 1 \\
0 & 0 & 1
\end{bmatrix}
(x^{0}_{13})
\Leftrightarrow
\begin{bmatrix}
1 & 0 & 0 \\
1 & 1 & 0 \\
0 & 1 & 1
\end{bmatrix}
(x^{2}_{1})
\end{align*}
\begin{align*}
\begin{bmatrix}
1 & 0 & 0 \\
1 & 1 & 0 \\
0 & 0 & 1
\end{bmatrix}
(x^{1}_{13})
\Leftrightarrow
\begin{bmatrix}
1 & 0 & 0 \\
0 & 1 & 1 \\
0 & 1 & 1
\end{bmatrix}
(x^{0}_{1})
\end{align*}

\subsection{Kernel}

After applying all these six reduction rules, the remaining types of items constitute the kernel of our problem. There are only a limited number of combinations of types that need to be considered and some cases can be handled by symmetry. Note that for each case, there are no more than 6 different types (i.e., at most 12 items) left.

By Reduction Rule~\ref{reduction:tuple}, at least one of $x_1$, $x_2$, and $x_3$ must be used up.
We consider the following cases.

\begin{itemize}
\item Only items of category $x_3$ are used up, while items of category $x_1$ and $x_2$ remain. In this case, $x_{13}$, $x_{23}$ and $x_{123}$ cannot exist due to Reduction Rule~\ref{reduction:tuple}.
  \begin{itemize}
  \item $|x_1| \geq |x_2| > 0$, $|x_{12}| \geq 0$, $|x_{0}| \geq 0$
  \end{itemize}
  There are four combinations of types due to Reduction Rule~\ref{reduction:row-sum-1}.
\item Items of category $x_2$ and $x_3$ are used up, while items of category $x_1$ remain.
Using Reduction Rule~\ref{reduction:tuple}, we can restrict our attention to the following cases (categories that do not appear have no items):
  \begin{itemize}
    \item $|x_1| > 0$, $|x_{12}| \geq 0$, $|x_{0}| \geq 0$ (covered by the previous case, as we will also check when $|x_2| = 0$)
    \item $|x_1| > 0$, $|x_{13}| \geq 0$, $|x_{0}| \geq 0$
    \item $|x_1| > 0$, $0 \leq |x_{23}| \leq 1$, $|x_{0}| \geq 0$
    \item $|x_1| > 0$, $0 \leq |x_{123}| \leq 1$, $|x_{0}| \geq 0$
  \end{itemize}
  There are three combinations of types for $x_1$ and $x_{13}$, four combinations for $x_1$ and $x_{23}$, and two combinations for $x_1$ and $x_{123}$  due to Reduction Rule~\ref{reduction:row-sum-1}.
  \item Items of category $x_1$, $x_2$ and $x_3$ are used up, while items of category $x_0$ remain:
  \begin{itemize}
    \item $|x_0| > 0$, $|x_{12}| + |x_{13}| + |x_{23}| + |x_{123}| \leq 2$
  \end{itemize}
  Since two types of $x_{12}$ / $x_{13}$ / $x_{23}$ cannot coexist with each other, 
  there are eight combinations of types and we do not write them down explicitly.
\end{itemize}

\begin{proposition}
\label{prop:envy_simultaneously}
For the problem of fair division of indivisible items among three agents under the no-chore assumption and binary valuations, if Reduction Rules 1--6 cannot be applied, then there exists an EF1 allocation in which no pair of agents envy each other simultaneously.
\end{proposition}

\begin{proof}
We wrote a program to prove Proposition~\ref{prop:envy_simultaneously} by exhaustive search. There are 21 possible combinations of types which are listed in the end of this paper (8 cases for $x_0$ are omitted). There are no more than 12 items for each combination, thus we can verify that there always exists an EF1 allocation in which no pair of agents envy each other simultaneously. See the code appendix for more details.
\end{proof}

\subsection{Proof of Theorem~\ref{theorem:EF1}}

\begin{proof}
Given an instance, we first classify all items based on their corresponding types and calculate the number of items of each type. Then we can apply the six reduction rules in polynomial time. By Proposition~\ref{prop:envy_simultaneously}, we know the kernel always admits an EF1 allocation $\pi$ in which no pair of agents envy each other simultaneously. 

We set aside some items during the process of Reduction Rule~\ref{reduction:diagonal-0} and~\ref{reduction:row-1}. 
From the proofs of these reduction rules, these items can be allocated in such a way that the resulting allocation is EF1.
\end{proof}

\subsection{$x_1$, $x_2$ and $x_{12}$ I}

\begin{itemize}
\item $|x_{12}| \geq 0$ 
\end{itemize}

\begin{align*}
&
\begin{bmatrix}
1 & 0 & 0 \\
0 & 1 & 0 \\
1 & 0 & 1
\end{bmatrix}
x_{12}^1
\end{align*}

\begin{itemize}
\item $|x_{1}|> 0$
\end{itemize}

\begin{align*}
&
\begin{bmatrix}
1 & 0 & 0 \\
0 & 1 & 1 \\
1 & 0 & 1
\end{bmatrix}
x_{1}^1
\quad
\begin{bmatrix}
1 & 0 & 0 \\
1 & 1 & 0 \\
0 & 1 & 1
\end{bmatrix}
x_{1}^2
\end{align*}

\begin{itemize}
\item $|x_{2}| > 0$
\end{itemize}

\begin{align*}
&
\begin{bmatrix}
1 & 1 & 0 \\
0 & 1 & 0 \\
1 & 0 & 1
\end{bmatrix}
x_{2}^2
\end{align*}

\begin{itemize}
\item $|x_0| \geq 0$ 
\end{itemize}

$$
\begin{bmatrix}
1 & 0 & 1 \\
1 & 1 & 0 \\
0 & 1 & 1
\end{bmatrix}
\quad
\begin{bmatrix}
1 & 1 & 0 \\
0 & 1 & 1 \\
1 & 0 & 1
\end{bmatrix}
$$

\subsection{$x_1$, $x_2$ and $x_{12}$ II}

\begin{itemize}
\item $|x_{12}| \geq 0$ 
\end{itemize}

\begin{align*}
&
\begin{bmatrix}
1 & 0 & 0 \\
0 & 1 & 0 \\
1 & 0 & 1
\end{bmatrix}
x_{12}^1
\end{align*}

\begin{itemize}
\item $|x_{1}|> 0$
\end{itemize}

\begin{align*}
&
\begin{bmatrix}
1 & 0 & 0 \\
0 & 1 & 1 \\
1 & 0 & 1
\end{bmatrix}
x_{1}^1
\quad
\begin{bmatrix}
1 & 0 & 0 \\
1 & 1 & 0 \\
0 & 1 & 1
\end{bmatrix}
x_{1}^2
\end{align*}

\begin{itemize}
\item $|x_{2}| > 0$
\end{itemize}

\begin{align*}
&
\begin{bmatrix}
1 & 0 & 1 \\
0 & 1 & 0 \\
1 & 0 & 1
\end{bmatrix}
x_{2}^1
\end{align*}

\begin{itemize}
\item $|x_0| \geq 0$ 
\end{itemize}

$$
\begin{bmatrix}
1 & 0 & 1 \\
1 & 1 & 0 \\
0 & 1 & 1
\end{bmatrix}
\quad
\begin{bmatrix}
1 & 1 & 0 \\
0 & 1 & 1 \\
1 & 0 & 1
\end{bmatrix}
$$

\subsection{$x_1$, $x_2$ and $x_{12}$ III}

\begin{itemize}
\item $|x_{12}| \geq 0$ 
\end{itemize}

\begin{align*}
&
\begin{bmatrix}
1 & 0 & 0 \\
0 & 1 & 0 \\
0 & 1 & 1
\end{bmatrix}
x_{12}^0
\end{align*}

\begin{itemize}
\item $|x_{1}|> 0$
\end{itemize}

\begin{align*}
&
\begin{bmatrix}
1 & 0 & 0 \\
0 & 1 & 1 \\
0 & 1 & 1
\end{bmatrix}
x_{1}^0
\end{align*}

\begin{itemize}
\item $|x_{2}| > 0$
\end{itemize}

\begin{align*}
&
\begin{bmatrix}
1 & 0 & 1 \\
0 & 1 & 0 \\
0 & 1 & 1
\end{bmatrix}
x_{2}^0
\quad
\begin{bmatrix}
1 & 1 & 0 \\
0 & 1 & 0 \\
1 & 0 & 1
\end{bmatrix}
x_{2}^2
\end{align*}

\begin{itemize}
\item $|x_0| \geq 0$ 
\end{itemize}

$$
\begin{bmatrix}
1 & 0 & 1 \\
1 & 1 & 0 \\
0 & 1 & 1
\end{bmatrix}
\quad
\begin{bmatrix}
1 & 1 & 0 \\
0 & 1 & 1 \\
1 & 0 & 1
\end{bmatrix}
$$

\subsection{$x_1$, $x_2$ and $x_{12}$ IV}

\begin{itemize}
\item $|x_{12}| \geq 0$ 
\end{itemize}

\begin{align*}
&
\begin{bmatrix}
1 & 0 & 0 \\
0 & 1 & 0 \\
0 & 1 & 1
\end{bmatrix}
x_{12}^0
\end{align*}

\begin{itemize}
\item $|x_{1}|> 0$
\end{itemize}

\begin{align*}
&
\begin{bmatrix}
1 & 0 & 0 \\
1 & 1 & 0 \\
0 & 1 & 1
\end{bmatrix}
x_{1}^2
\end{align*}

\begin{itemize}
\item $|x_{2}| > 0$
\end{itemize}

\begin{align*}
&
\begin{bmatrix}
1 & 0 & 1 \\
0 & 1 & 0 \\
0 & 1 & 1
\end{bmatrix}
x_{2}^0
\quad
\begin{bmatrix}
1 & 1 & 0 \\
0 & 1 & 0 \\
1 & 0 & 1
\end{bmatrix}
x_{2}^2
\end{align*}

\begin{itemize}
\item $|x_0| \geq 0$ 
\end{itemize}

$$
\begin{bmatrix}
1 & 0 & 1 \\
1 & 1 & 0 \\
0 & 1 & 1
\end{bmatrix}
\quad
\begin{bmatrix}
1 & 1 & 0 \\
0 & 1 & 1 \\
1 & 0 & 1
\end{bmatrix}
$$

\subsection{$x_1$ and $x_{13}$ I}

\begin{itemize}
\item $|x_{13}| \geq 0$ 
\end{itemize}

\begin{align*}
&
\begin{bmatrix}
1 & 0 & 0 \\
0 & 1 & 1 \\
0 & 0 & 1
\end{bmatrix}
x_{13}^0
\end{align*}

\begin{itemize}
\item $|x_{1}|> 0$
\end{itemize}

\begin{align*}
&
\begin{bmatrix}
1 & 0 & 0 \\
0 & 1 & 1 \\
1 & 0 & 1
\end{bmatrix}
x_{1}^1
\end{align*}

\begin{itemize}
\item $|x_0| \geq 0$ 
\end{itemize}

$$
\begin{bmatrix}
1 & 0 & 1 \\
1 & 1 & 0 \\
0 & 1 & 1
\end{bmatrix}
\quad
\begin{bmatrix}
1 & 1 & 0 \\
0 & 1 & 1 \\
1 & 0 & 1
\end{bmatrix}
$$

\subsection{$x_1$ and $x_{13}$ II}

\begin{itemize}
\item $|x_{13}| \geq 0$ 
\end{itemize}

\begin{align*}
&
\begin{bmatrix}
1 & 0 & 0 \\
0 & 1 & 1 \\
0 & 0 & 1
\end{bmatrix}
x_{13}^0
\end{align*}

\begin{itemize}
\item $|x_{1}|> 0$
\end{itemize}

\begin{align*}
&
\begin{bmatrix}
1 & 0 & 0 \\
0 & 1 & 1 \\
0 & 1 & 1
\end{bmatrix}
x_{1}^0
\end{align*}

\begin{itemize}
\item $|x_0| \geq 0$ 
\end{itemize}

$$
\begin{bmatrix}
1 & 0 & 1 \\
1 & 1 & 0 \\
0 & 1 & 1
\end{bmatrix}
\quad
\begin{bmatrix}
1 & 1 & 0 \\
0 & 1 & 1 \\
1 & 0 & 1
\end{bmatrix}
$$

\subsection{$x_1$ and $x_{13}$ III}
\begin{itemize}
\item $|x_{13}| \geq 0$ 
\end{itemize}

\begin{align*}
&
\begin{bmatrix}
1 & 0 & 0 \\
1 & 1 & 0 \\
0 & 0 & 1
\end{bmatrix}
x_{13}^1
\end{align*}

\begin{itemize}
\item $|x_{1}|> 0$
\end{itemize}

\begin{align*}
&
\begin{bmatrix}
1 & 0 & 0 \\
0 & 1 & 1 \\
1 & 0 & 1
\end{bmatrix}
x_{1}^1
\quad
\begin{bmatrix}
1 & 0 & 0 \\
1 & 1 & 0 \\
0 & 1 & 1
\end{bmatrix}
x_{1}^2
\end{align*}

\begin{itemize}
\item $|x_0| \geq 0$ 
\end{itemize}

$$
\begin{bmatrix}
1 & 0 & 1 \\
1 & 1 & 0 \\
0 & 1 & 1
\end{bmatrix}
\quad
\begin{bmatrix}
1 & 1 & 0 \\
0 & 1 & 1 \\
1 & 0 & 1
\end{bmatrix}
$$

\subsection{$x_1$ and $x_{23}$   I}

\begin{itemize}
\item $|x_{23}| \leq 1$ 
\end{itemize}

\begin{align*}
&
\begin{bmatrix}
1 & 1 & 0 \\
0 & 1 & 0 \\
0 & 0 & 1
\end{bmatrix}
x_{23}^1
\end{align*}

\begin{itemize}
\item $|x_{1}|> 0$
\end{itemize}

\begin{align*}
&
\begin{bmatrix}
1 & 0 & 0 \\
0 & 1 & 1 \\
1 & 0 & 1
\end{bmatrix}
x_{1}^1
\quad
\begin{bmatrix}
1 & 0 & 0 \\
1 & 1 & 0 \\
0 & 1 & 1
\end{bmatrix}
x_{1}^2
\end{align*}

\begin{itemize}
\item $|x_0| \leq 1$ 
\end{itemize}

$$
\begin{bmatrix}
1 & 0 & 1 \\
1 & 1 & 0 \\
0 & 1 & 1
\end{bmatrix}
\quad
\begin{bmatrix}
1 & 1 & 0 \\
0 & 1 & 1 \\
1 & 0 & 1
\end{bmatrix}
$$

\subsection{$x_1$ and $x_{23}$   II}

\begin{itemize}
\item $|x_{23}| \leq 1$ 
\end{itemize}

\begin{align*}
&
\begin{bmatrix}
1 & 0 & 1 \\
0 & 1 & 0 \\
0 & 0 & 1
\end{bmatrix}
x_{23}^0
\end{align*}

\begin{itemize}
\item $|x_{1}|> 0$
\end{itemize}

\begin{align*}
&
\begin{bmatrix}
1 & 0 & 0 \\
0 & 1 & 1 \\
1 & 0 & 1
\end{bmatrix}
x_{1}^1
\quad
\begin{bmatrix}
1 & 0 & 0 \\
1 & 1 & 0 \\
0 & 1 & 1
\end{bmatrix}
x_{1}^2
\end{align*}

\begin{itemize}
\item $|x_0| \geq 0$ 
\end{itemize}

$$
\begin{bmatrix}
1 & 0 & 1 \\
1 & 1 & 0 \\
0 & 1 & 1
\end{bmatrix}
\quad
\begin{bmatrix}
1 & 1 & 0 \\
0 & 1 & 1 \\
1 & 0 & 1
\end{bmatrix}
$$

\subsection{$x_1$ and $x_{23}$   III}

\begin{itemize}
\item $|x_{23}| \leq 1$ 
\end{itemize}

\begin{align*}
&
\begin{bmatrix}
1 & 1 & 0 \\
0 & 1 & 0 \\
0 & 0 & 1
\end{bmatrix}
x_{23}^1
\end{align*}

\begin{itemize}
\item $|x_{1}|> 0$
\end{itemize}

\begin{align*}
&
\begin{bmatrix}
1 & 0 & 0 \\
0 & 1 & 1 \\
0 & 1 & 1
\end{bmatrix}
x_{1}^0
\end{align*}

\begin{itemize}
\item $|x_0| \geq 0$ 
\end{itemize}

$$
\begin{bmatrix}
1 & 0 & 1 \\
1 & 1 & 0 \\
0 & 1 & 1
\end{bmatrix}
\quad
\begin{bmatrix}
1 & 1 & 0 \\
0 & 1 & 1 \\
1 & 0 & 1
\end{bmatrix}
$$

\subsection{$x_1$ and $x_{23}$   IV}

\begin{itemize}
\item $|x_{23}| \leq 1$ 
\end{itemize}

\begin{align*}
&
\begin{bmatrix}
1 & 0 & 1 \\
0 & 1 & 0 \\
0 & 0 & 1
\end{bmatrix}
x_{23}^0
\end{align*}

\begin{itemize}
\item $|x_{1}|> 0$
\end{itemize}

\begin{align*}
&
\begin{bmatrix}
1 & 0 & 0 \\
0 & 1 & 1 \\
0 & 1 & 1
\end{bmatrix}
x_{1}^0
\end{align*}

\begin{itemize}
\item $|x_0| \geq 0$ 
\end{itemize}

$$
\begin{bmatrix}
1 & 0 & 1 \\
1 & 1 & 0 \\
0 & 1 & 1
\end{bmatrix}
\quad
\begin{bmatrix}
1 & 1 & 0 \\
0 & 1 & 1 \\
1 & 0 & 1
\end{bmatrix}
$$

\subsection{$x_1$ and $x_{123}$   I}

\begin{itemize}
\item $|x_{123}| \leq 1$ 
\end{itemize}

\begin{align*}
&
\begin{bmatrix}
1 & 0 & 0 \\
0 & 1 & 0 \\
0 & 0 & 1
\end{bmatrix}
x_{123}^0
\end{align*}

\begin{itemize}
\item $|x_{1}|> 0$
\end{itemize}

\begin{align*}
&
\begin{bmatrix}
1 & 0 & 0 \\
0 & 1 & 1 \\
1 & 0 & 1
\end{bmatrix}
x_{1}^1
\quad
\begin{bmatrix}
1 & 0 & 0 \\
1 & 1 & 0 \\
0 & 1 & 1
\end{bmatrix}
x_{1}^2
\end{align*}

\begin{itemize}
\item $|x_0| \geq 0$ 
\end{itemize}

$$
\begin{bmatrix}
1 & 0 & 1 \\
1 & 1 & 0 \\
0 & 1 & 1
\end{bmatrix}
\quad
\begin{bmatrix}
1 & 1 & 0 \\
0 & 1 & 1 \\
1 & 0 & 1
\end{bmatrix}
$$

\subsection{$x_1$ and $x_{123}$   II}

\begin{itemize}
\item $|x_{123}| \leq 1$ 
\end{itemize}

\begin{align*}
&
\begin{bmatrix}
1 & 0 & 0 \\
0 & 1 & 0 \\
0 & 0 & 1
\end{bmatrix}
x_{123}^0
\end{align*}

\begin{itemize}
\item $|x_{1}|> 0$
\end{itemize}

\begin{align*}
&
\begin{bmatrix}
1 & 0 & 0 \\
0 & 1 & 1 \\
0 & 1 & 1
\end{bmatrix}
x_{1}^0
\end{align*}

\begin{itemize}
\item $|x_0| \geq 0$ 
\end{itemize}

$$
\begin{bmatrix}
1 & 0 & 1 \\
1 & 1 & 0 \\
0 & 1 & 1
\end{bmatrix}
\quad
\begin{bmatrix}
1 & 1 & 0 \\
0 & 1 & 1 \\
1 & 0 & 1
\end{bmatrix}
$$

\end{document}